\documentclass[a4paper,twocolumn,11pt, accepted = 2021-08-29]{quantumarticle}
\pdfoutput=1
\usepackage[utf8]{inputenc}
\usepackage[english]{babel}
\usepackage[T1]{fontenc}
\usepackage{amsmath}
\usepackage{hyperref}

\usepackage[dvipsnames,x11names]{xcolor} 				 
\usepackage{soul}
\usepackage{float}
\usepackage{tabularx}
\usepackage{graphicx}
\usepackage{bm}
\usepackage{enumerate}
\usepackage{enumitem}
\usepackage[numbers,sort&compress]{natbib} 	
\usepackage{graphicx}
\usepackage{amsmath}
\usepackage{amsthm}
\usepackage{amssymb}
\usepackage{mathtools}
\usepackage{dsfont}

\theoremstyle{plain}
\newtheorem{thm}{Theorem}[section]
\newtheorem{lem}[thm]{Lemma}

\theoremstyle{definition}
\newtheorem{defn}{Definition}[section]

\theoremstyle{remark}

\newcommand{\tr}{\mathrm{tr}}
\newcommand{\identity}{\mathds{1}}
\newcommand{\ket}[1]{\vert{#1}\rangle}
\newcommand{\keket}[1]{\vert{#1}\rangle\hspace{-1mm}\rangle}
\newcommand{\bra}[1]{\langle{#1}\vert}
\newcommand{\brabra}[1]{\langle\hspace{-1mm}\langle{#1}\vert}
\newcommand{\ketbra}[2]{\ket{#1}\hspace{-1mm}\bra{#2}}
\newcommand{\keketbra}[2]{\keket{#1}\hspace{-1mm}\brabra{#2}}
\newcommand{\tracerep}[2]{\prescript{}{#1}{#2}}
\newcommand{\aux}{\textrm{aux}}

\newcommand{\AB}{A \rightarrow B}

\newcommand{\ABC}{A \rightarrow B \rightarrow C}

\newcommand{\ucnot}{\textrm{U}_{\textrm{CNOT}}}
\newcommand{\uswap}{\textrm{U}_{\textrm{SWAP}}}
\newcommand{\ups}{\textrm{U}_{\textrm{PS}}}
\newcommand{\Wcc}{W_{\textrm{CC}}}
\newcommand{\Wdc}{W_{\textrm{DC}}}

\newcommand{\cyan}[1]{{#1}} 
\newcommand{\red}[1]{\textcolor{red}{#1}} 
\newcommand{\green}[1]{\textcolor{green}{#1}}

\makeatletter 
\newcommand{\doublewidetilde}[1]{{%
  \mathpalette\double@widetilde{#1}%
}}
\newcommand{\double@widetilde}[2]{%
  \sbox\z@{$\m@th#1\widetilde{#2}$}%
  \ht\z@=.9\ht\z@
  \widetilde{\box\z@}%
}

\newcommand{\map}[1]{\widetilde{#1}}  

\makeatother

\begin{document}

\title{Simple and maximally robust processes with no classical common-cause or direct-cause explanation}

\author{Marcello Nery}
\email{mnerygvb@gmail.com}
\affiliation{Departamento de Física, Universidade Federal de Minas Gerais, Av. Pres. Antonio Carlos 6627 - Belo Horizonte, MG, Brazil - 31270-901.}

\author{Marco Túlio Quintino}
 \orcid{0000-0003-1332-3477}
\affiliation{Institute for Quantum Optics and Quantum Information (IQOQI), Austrian
Academy of Sciences, Boltzmanngasse 3, A-1090 Vienna, Austria}
\affiliation{Faculty of Physics, University of Vienna, Boltzmanngasse 5, 1090 Vienna, Austria}

\author{Philippe Allard Guérin}
\affiliation{Faculty of Physics, University of Vienna, Boltzmanngasse 5, 1090 Vienna, Austria}
\affiliation{Institute for Quantum Optics and Quantum Information (IQOQI), Austrian
Academy of Sciences, Boltzmanngasse 3, A-1090 Vienna, Austria}
\affiliation{Perimeter Institute for Theoretical Physics, 31 Caroline St. N, Waterloo, Ontario, N2L 2Y5, Canada}

\author{Thiago O. Maciel}
\affiliation{Departamento de Física, Universidade Federal de Santa Catarina, Florianópolis, SC, 88040-900, Brazil}
\affiliation{Universidade Federal do Rio de Janeiro, Caixa Postal 68528, Rio de Janeiro, RJ 21941-972, Brazil}

\author{Reinaldo O. Vianna}
\affiliation{Departamento de Física, Universidade Federal de Minas Gerais, Av. Pres. Antonio Carlos 6627 - Belo Horizonte, MG, Brazil - 31270-901.}

\date{September 02, 2021}

\maketitle

\begin{abstract}
  Guided by the intuition of coherent superposition of causal relations, recent works presented quantum processes without classical common-cause and direct-cause explanation, that is, processes which cannot be written as probabilistic mixtures of quantum common-cause and quantum direct-cause relations (CCDC). In this work, we analyze the minimum requirements for a quantum process to fail to admit a CCDC explanation and present ``simple'' processes, which we prove to be the most robust ones against general noise. These simple processes can be realized by preparing a maximally entangled state and applying the identity quantum channel, thus not requiring an explicit coherent mixture of common-cause and direct-cause, exploiting the possibility of a process to have both relations simultaneously. 
We then prove that, although all bipartite direct-cause processes are bipartite separable operators, there exist bipartite separable processes which are not direct-cause. This shows that the problem of deciding \cyan{whether} a process is direct-cause \emph{is not} equivalent to entanglement certification and points out the limitations of entanglement methods to detect non-classical CCDC processes. We also present a semi-definite programming hierarchy that can detect and quantify the non-classical CCDC robustnesses of every non-classical CCDC process.
   Among other results, our numerical methods allow us to show that the simple processes presented here are likely to be also the maximally robust against white noise. Finally, we explore the equivalence between bipartite direct-cause processes and bipartite processes without quantum memory, to present a separable process which cannot be realized as a process without quantum memory. 
\end{abstract}

\section*{Introduction}

	Common-cause and direct-cause relations are the building blocks of classical causal models, as a causal model of multiple variables consists of combinations of common-cause and direct-cause relations between them. Understanding the relationship between cause and effects is one of the fundamental goals of several physical theories and of statistical analysis \cite{reichenbach1991direction}. Also, determining the causal relation behind a correlation of two objects is a fundamental problem in causal inference theory and plays a main role in topics such as hypothesis testing, social sciences, medicine, and machine learning \cite{Pearl2000,illari2011causality}.
In order to analyze causality in quantum phenomena, recent works proposed new frameworks of theory of Bayesian inference \cite{Leifer2011} and causal modeling \cite{Costa2016}, where understanding common-cause and direct-cause relations in quantum mechanics plays a fundamental role.

Quantum causal processes consist of a sequence of quantum operations and may be analyzed from different equivalent perspectives, such as sequential quantum operations via non-Markovian processes \cite{Modi2012,Pollock2015,Li2018,Milz2020}, fragments of quantum circuits via quantum combs \cite{chiribella08,Chiribella2009a}, quantum channels with memory \cite{kretschmann05}, and quantum strategies for playing an $n$-turn game \cite{Gutoski2007}. 
	When referring to causal modeling, a process which can be written as a probabilistic mixture of common-cause (CC) and direct-cause (DC) processes is said to admit a classical common-cause or direct-cause (CCDC) explanation, since the causal relations in such process could be simulated classically, by sampling from a probability distribution and then implementing either a common-cause or direct-cause process.
		
	 In Ref.\,\cite{Ried2015} the authors consider the different possibilities of combining CC and DC processes and, inspired by coherent mixture of quantum channels, the authors of Ref.\,\cite{MacLean2016} experimentally certify the existence of a quantum process with no classical CCDC explanation. Also, inspired by the \emph{Quantum Switch}\,\cite{Chiribella2009,Procopio2014,Rubino2016}, Ref.\,\cite{Feix2016} presents a coherent superposition of common-cause and direct-cause processes which cannot be explained by a classical CCDC. 

	In this work, we analyze quantum processes which cannot be decomposed as probabilistic mixtures of common-cause and direct-cause ones, hence admitting no classical CCDC explanation. We focus on the bipartite case, the simplest scenario where such processes can exist. In this simplest scenario, exploring the possibility of a process to have both CC and DC relations simultaneously, we present a process which can be simply realized by preparing a maximally entangled state and an identity channel, not needing an interpretation of it as a coherent mixture of causal relations. When attempting to work with the minimum non-trivial dimensions, we propose another process, which requires a controlled-NOT operation. We prove that these processes are the most robust ones against their worst possible noise, known as generalized noise. We also develop a semi-definite programming (SDP) numerical approach to quantify the non-classical CCDC property of a process, based on its robustness against white and general noise. Our numerical methods allow us to show that these simple processes presented here are also maximally robust against white noise when considering qubits, and to quantify the non-classical CCDC property in any bipartite ordered process.

	We also investigate the differences and similarities between quantum entanglement and processes without a direct-cause explanation, showing that there exist separable processes which are not direct-cause. This example points out the limitations of purely entanglement-based methods to detect non direct-cause processes and answers a conjecture first raised in Ref.\,\cite{Giarmatzi2018}. 
	
	Finally, we connect our results with a different related field by proving that for the bipartite case, processes without quantum memory \cite{Giarmatzi2018} are equivalent to processes having a direct-cause decomposition. With that, we present a bipartite separable process which cannot be realized as a process without quantum memory and contribute to the understanding of the relation between entanglement and of quantum memory \cite{Milz2017,Giarmatzi2018,Taranto2019,Taranto2020,Milz2020A}.
	
\section{Mathematical Preliminaries}
\label{sec:mp}

\subsection{The Choi-Jamio\l{}kowski isomorphism}

As is standard in several branches of quantum information theory, we will make use of the Choi-Jamio\l{}kowski isomorphism to represent linear operators and linear maps \cite{depillis67,jamiolkowski72,choi75}. We first present the ``pure'' version of the isomorphism, which allows us to represent linear operators as bipartite vectors. Here $A_I$ and $A_O$ are finite-dimensional complex vector spaces, which will late be identified with Alice's input and Alice's output, respectively.
\begin{defn}[``Pure'' Choi-Jamio\l{}kowski isomorphism] 
	Let $\textrm{U}: (A_I) \rightarrow (A_O)$ be a linear transformation (which is not necessarily unitary). The Choi vector $\keket{\textrm{U}}\in A_I\otimes A_O$ of the operator $\textrm{U}$ is defined as 
	\begin{equation}	
	\keket{\textrm{U}} := \sum_{i} \ket{i} \otimes (\textrm{U} \, \ket{i}),
	\end{equation}
where $\{\ket{i}\}_{i=0}^{d_{A_I}-1}$ is the computational basis for $A_I$ and $d_{A_I}$ is the dimension of $A_I$
\end{defn}

In quantum information theory, linear transformations between operators are sometimes referred to as linear maps, or simply, as maps. In this paper, we identify linear maps with a single tilde. For instance, let $\mathcal{L}(A_I)$ be the set of linear operators mapping $A_I$ to itself, $\map{\Lambda}:\mathcal{L}(A_I)\to\mathcal{L}(A_O)$ is a map transforming operators from $\mathcal{L}(A_I)$ to $\mathcal{L}(A_O)$. We now present the standard version of the Choi-Jamio\l{}kowski isomorphism, which allows us to represent linear maps as bipartite operators.
\begin{defn}[Choi-Jamio\l{}kowski isomorphism] 
	Let $\widetilde{\Lambda}: \mathcal{L}(A_I) \rightarrow \mathcal{L}(A_O)$ be a linear map and $\{\ket{i}\}_{i=0}^{d_{A_I}-1}$  be the computational basis for $A_I$. The Choi operator $\Lambda\in \mathcal{L}(A_I\otimes A_O)$ of the map $\map{\Lambda}$ is defined as 
	\begin{equation}
	\label{eqchoimatrix}	
	\Lambda := \sum_{ij} \ketbra{i}{j}\otimes \map{\Lambda}(\ketbra{i}{j}).
	\end{equation}
\end{defn}

	A common equivalent way to define the Choi operator of a map $\map{\Lambda}$ is given by the equation
\begin{equation}
	\label{eqchoimatrix2}	
	\Lambda := \left[\left(\widetilde{\identity} \otimes \widetilde{\Lambda}\right)(\keketbra{\identity}{\identity}) \right],	
	\end{equation}
where $\widetilde{\identity} : \mathcal{L}(A_I) \rightarrow \mathcal{L}(A_I)$ is an identity channel, and $\keket{\identity} = \sum_{i=0}^{d_{A_I}-1} \ket{ii} \in A_I \otimes A_I$.
	
	The action of the map $\widetilde{\Lambda}: \mathcal{L}(A_I) \rightarrow \mathcal{L}(A_O)$  on an operator $\rho\in \mathcal{L}(A_I) $ can be obtained from the Choi operator $\Lambda$ by means of the relation
\begin{equation}
\label{eqmapfromchoi}
\widetilde{\Lambda}(\rho) = \tr_{A_I} \left[ ( {\rho^{A_I}}^T \otimes \identity^{A_O}) \cdot \Lambda \right],
\end{equation}
where $\identity$ is the identity operator and $(\cdot)^T$ stands for the transposition in the computational basis.

	A linear map is a quantum channel iff it is completely positive (CP) and trace-preserving (TP). In the Choi-Jamio\l{}kowski representation, a map $\widetilde{\Lambda}: \mathcal{L}(A_I) \rightarrow \mathcal{L}(A_O)$  is CP iff $\Lambda \succeq 0$, \emph{i.e.}, $\Lambda$ is positive semi-definite, and TP iff $\tr_{A_O}(\Lambda)=\identity^{A_I}$.

For the sake of clarity, some equations in this manuscript will explicitly identify the input and output Hilbert spaces of a map with a superscript in the symbol representing the Choi operator of the map. In the above situation, for example, the Choi operator is equivalently represented by $\Lambda^{A_I / A_O}$, indicating that the map $\map{\Lambda}$ takes an operator from $\mathcal{L}(A_I)$ to one in $\mathcal{L}(A_O)$.

\subsection{The link product operation}

First defined in Ref.\,\cite{Chiribella2009a}, the link product is a useful mathematical tool to deal with composition of elements in a quantum circuit presented in the Choi-Jamio\l{}kowski representation.

\begin{defn}[Link product]
Let $\Lambda_1 \in \mathcal{L}( A\otimes A') $ and $\Lambda_2 \in \mathcal{L}( A'\otimes A'') $ be linear operators. The link product between $\Lambda_1$ and $\Lambda_2$ is defined as
    \begin{align}
	    \begin{split}
	        \label{eq:linkprod}
	        \Lambda_2 * \Lambda_1 := \tr_{A'}  \Bigg[  \bigg( \Lambda_1^{T_{A'}}\ \otimes \identity^{A''}\bigg)  
	          \cdot  \bigg(\identity^{A} \otimes \Lambda_2 \bigg)\Bigg],
		\end{split}    
    \end{align}
where $\Lambda_1^{T_{A'}}$ is the partial transpose of $\Lambda_1$ on the space $A'$.
\end{defn}

	As stated previously, the link product is useful for composing linear maps and quantum objects. If $\map{\Lambda}_1: \mathcal{L}(A) \rightarrow \mathcal{L}(A')$ and $\map{\Lambda}_2: \mathcal{L}(A') \rightarrow \mathcal{L}(A'')$ are linear maps with Choi operators $\Lambda_1 \in \mathcal{L}(A \otimes A')$ and $\Lambda_2 \in \mathcal{L}(A' \otimes A'')$, respectively, it can be shown that the Choi operator of the composition $\widetilde{\Phi} := \map{\Lambda}_2 \circ \map{\Lambda}_1: \mathcal{L}(A) \rightarrow \mathcal{L}(A'')$ is 
$   \Phi = \Lambda_2 * \Lambda_1.$

In particular, when $\rho \in  \mathcal{L}(A)$, is a linear operator with no components on $A'$, and $\Lambda \in \mathcal{L}(A \otimes A')$, we have
\begin{align}
	\begin{split}
	\Lambda * \rho:=& \tr_{A}  \left[  \left( \rho^{T_{A}}\ \otimes \identity^{A'}\right)  \cdot  \left(\Lambda \right)\right] \\
			=&\widetilde{\Lambda}(\rho).
	\end{split}
\end{align}

Hence, we can use the link product to represent quantum operations being performed on quantum states. 
Also, if  $\rho_1 \in \mathcal{L}(A_1)$ and $\rho_2 \in \mathcal{L}(A_2)$, acting on different spaces, we have $\rho_1 * \rho_2 := \rho_1 \otimes \rho_2$. Therefore, the link product of independent systems is simply the tensor product. 

Additionally, when $\rho, M \in \mathcal{L}(A)$ act in the same linear space, the link product is given by $\rho * M := \tr(M^{\textrm{T}} \cdot \rho ) $, which is simply the trace of their product with an extra transposition. This form can also be used to write Born's rule.

\section{The CCDC scenario}
\label{sec:ccdcscen}

The main objects analyzed in this work are the bipartite ordered processes, which may be understood as a physical dynamics which flows from Alice to Bob. 
	For all definitions in this section, we consider a scenario where Alice is a quantum physicist in a laboratory who can perform any quantum operation that transforms states from system $A_I$ (Alice input) to $A_O$ (Alice output), and Bob is a quantum physicist in a laboratory who can perform any quantum measurement on states defined on system $B_I$ (Bob input).

\subsection{Markovian processes}
\label{subse:Markovian}

	We start by presenting the definition of bipartite \textit{Markovian processes}, which admits a \textit{direct-cause} interpretation \cite{Modi2012,Milz2017a,Pollock2018} and are also known in the literature as quantum processes with no memory \cite{Giarmatzi2018}, and cause-effect causal maps \cite{Ried2015}. In a Markovian scenario, Alice receives a known quantum state $\rho\in\mathcal{L}(A_I)$ which acts on her input system space. Alice can then perform an arbitrary quantum operation%
	\footnote{Which may be a quantum channel (deterministic quantum operation) or a quantum instrument (probabilistic quantum operation).}%
	$\map{\Lambda}:\mathcal{L}(A_I)\to \mathcal{L}(A_O)$ to obtain a quantum state on system $A_O$, which will then be subjected to a known deterministic dynamics described by a channel $\map{D}:\mathcal{L}(A_O)\to \mathcal{L}(B_I)$. Finally, the state
	 $\map{D}\Big(\map{\Lambda}(\rho)\Big)\in\mathcal{L}(B_I)$ arrives to Bob, who can perform an arbitrary measurement on it. See Fig.\,\ref{fig:circuitdc} for a circuit based pictorial illustration.

	A bipartite Markovian process can then be described by the known quantum state $\rho\in\mathcal{L}(A_I)$ and the known dynamics represented by the channel $\map{D}:\mathcal{L}(A_O)\to \mathcal{L}(B_I)$. In the Choi-Jamio\l{}kowski representation, this can be conveniently described by $W_{\textrm{Markov}} := \rho^{A_I} \otimes D^{A_O/B_I}$. Under this definition, for any quantum operation $\map{\Lambda}:\mathcal{L}(A_I)\to \mathcal{L}(A_O)$ performed by Alice, the quantum state arriving in Bob's input space can be obtained via the link product as
	\small
	\begin{align}
	\begin{split}
	 W_{\textrm{Markov}} * \Lambda^{A_I/A_O} &=\rho^{A_I} \otimes D^{A_O/B_I} * \Lambda^{A_I/A_O} \\
	 &=\rho^{A_I} * D^{A_O/B_I} * \Lambda^{A_I/A_O} \\
		&= D^{A_O/B_I}*\Lambda^{A_I/A_O} *\rho^{A_I} \\
		&=D^{A_O/B_I} * \map{\Lambda}(\rho^{A_I})  \\
		&= \map{D}\Big(\map{\Lambda}(\rho)\Big).
	\end{split}
\end{align}	 
\normalsize

\begin{figure}
    \centering
    \includegraphics[width=9cm]{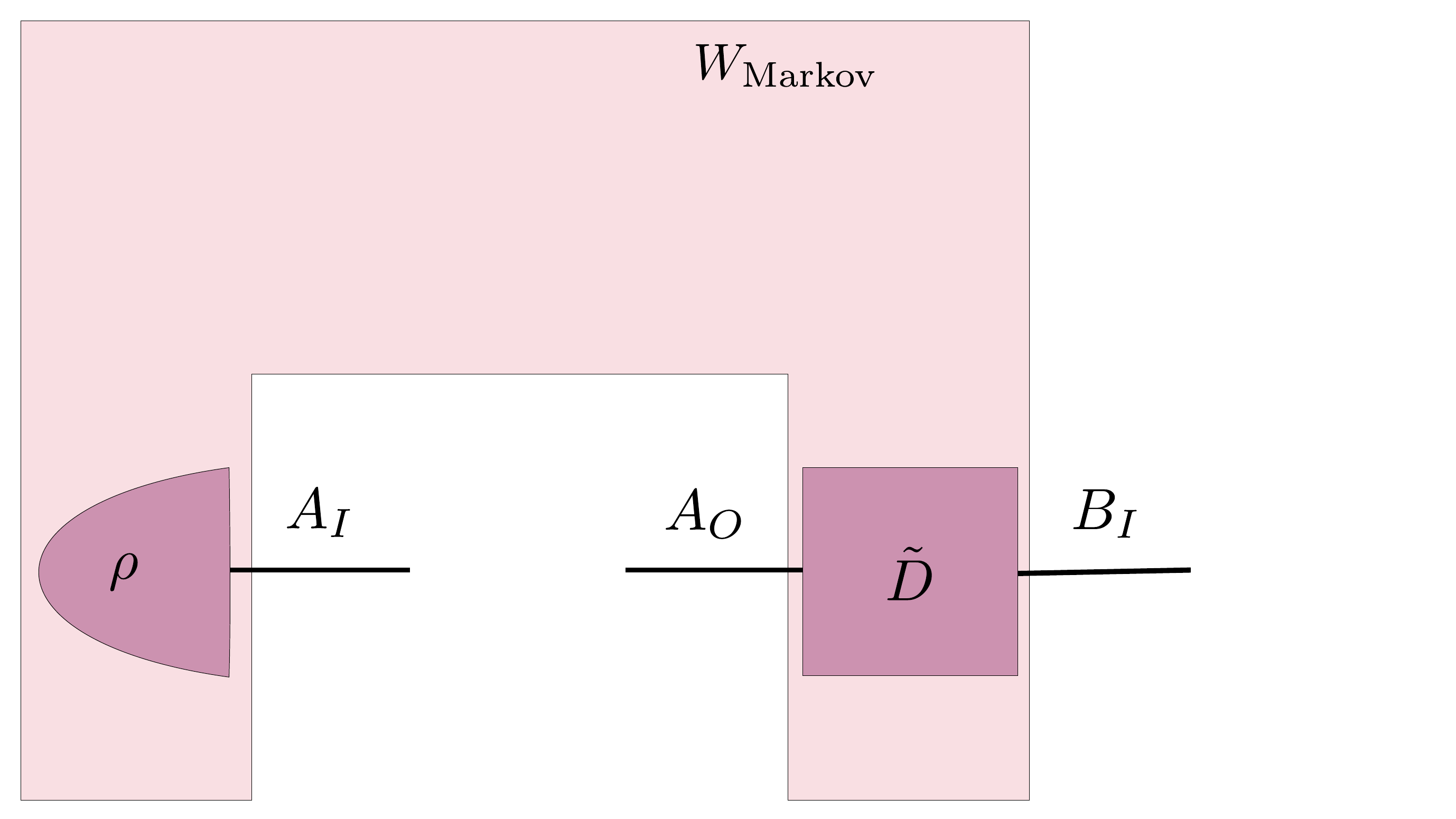}
    \caption{Circuit representation of Markovian processes: Alice receives a fixed state $\rho\in \mathcal{L}(H_{A_I})$ on which she can perform an arbitrary quantum operation. After Alice's operation, the system is subjected to a fixed dynamics described by a quantum channel   $\map{D}:\mathcal{L}(A_O)\to \mathcal{L}(B_I)$, then arriving at Bob's input space laboratory. No auxiliary system is used in this circuit, which implies that correlations between the parties are only due to the direct communication from Alice to Bob.}
    \label{fig:circuitdc}
\end{figure}

\begin{defn}[Markovian process]
	A linear operator $W_{\text{Markov}}\in \mathcal{L}(A_I \otimes A_O \otimes B_I)$ is bipartite Markovian process if it can be written as
\label{eq:edcproc} 
    \begin{align}
    \begin{split}
    W_{\textrm{Markov}}  :=& \, \rho^{A_I} * D^{A_O/B_I}, \\
    	              =& \, \rho^{A_I} \otimes D^{A_O/B_I},
	\end{split}    
    \end{align}

where $\rho^{A_I}$ is a quantum state \textit{i.e.,} $\rho\succeq0$, $\tr(\rho)=1$ and $D^{A_O/B_I}$ is the Choi operator of a quantum channel from the output of Alice to Bob's input \textit{i.e.,} $D\succeq0$, $\tr_{B_I}(D)=\identity_{A_O}$.     
We denote the set of Markovian processes by $\mathcal{L}_{\textrm{Markov}}$. 
\end{defn}

\subsection{Direct-cause processes}
\label{subse:dcprocesses}

	In a Markovian process, all correlations between Alice and Bob admit a direct-cause interpretation, since no auxiliary system or environment is required at any point. In a scenario where Alice and Bob may share prior \emph{classical} correlations, it is natural to define \textit{direct-cause} processes as any process which can be written as a probabilistic mixture of Markovian processes. Hence, the correlations arising from these processes can always be explained by classical mixture of direct-cause ones.

\begin{defn}[Direct-cause process] \label{def:dc}
A linear operator $W_{\text{DC}}\in \mathcal{L}(A_I \otimes A_O \otimes B_I)$ is direct-cause if it is a classical mixture of Markovian processes, that is, $W_{\text{DC}}$ can be written as
\begin{subequations}
    \begin{align}
        \Wdc  :=& \sum_i p_i \rho_i^{A_I} * D_i^{A_O/B_I}  \\ 
        =& \sum_i p_i \rho_i^{A_I} \otimes D_i^{A_O/B_I}, \ p_i \in [0,1], \label{eqdcproc}
    \end{align}
\end{subequations}
where $\rho^{A_I}_i$ are quantum states, \textit{i.e.,} $\rho_i\succeq 0$ and $\tr(\rho_i)=1 \ \forall i$, whereas $D^{A_O/B_I}_i$ are Choi operators of quantum channels from the output of Alice to Bob's input, \textit{i.e.,} $D_i\succeq 0$ and $\tr_{B_I}(D_i)=\identity_{A_O}\ \forall i$. We denote the set of direct-cause processes by $\mathcal{L}_{\textrm{DC}}$.
\end{defn}
 
In the bipartite scenario, an ordered process is direct-cause if and only if it is a process without quantum memory, as defined and shown in Ref.\,\cite{Giarmatzi2018}. See Appendix \ref{sec:classmemdc} for details. Although Ref.\,\cite{Feix2016} analyzes  tripartite ordered processes, when bipartite processes are considered, their definition is equivalent to the one presented in Refs.\,\cite{Ried2015, MacLean2016}, which is also equivalent to the one used here. We discuss this in more details in Section\,\ref{sec:feix}.

\subsubsection{Quantum separability as an outer approximation for the set of DC processes}
\label{subsubsec:entapprox}

 \cyan{A positive semi-definite linear operator $W \in \mathcal{L}(A_I \otimes A_O \otimes B_I)$ is separable in the bipartition $A_I|A_O B_I$ if there exist a probability distribution $p_i$ and positive semi-definite operators $\rho_i^{A_I} \succeq 0$ and $D_i^{A_O/B_I} \succeq 0$ such that $W = \sum_i p_i\rho_i^{A_I}\otimes D_i^{A_O/B_I}$. Which coincides with the definition of a DC process (Def.~\ref{def:dc}) without the trace normalization constraints.}
	It follows directly from Eq.~\eqref{eqdcproc} that every direct-cause process is separable in the bipartition $A_I|A_O B_I$.
	Conversely, if a process is entangled in the bipartition $A_I|A_OB_I$, such process cannot have a DC explanation. This allows us to employ techniques from entanglement theory to certify whether a process is not direct-cause. A simple outer approximation for the set of separable states is the set of states with positive partial transpose (PPT) \cite{Peres1996,Horodecki1996}. Also, it is known that the set of states with a PPT k-symmetric extension \cite{Doherty2001, Doherty2003} forms a hierarchy of sets which converges to the set of separable states when $k \to \infty$. For every fixed $k$, checking if a state has a PPT symmetric k-extension can be done via SDP. 
The precise definition of a PPT k-symmetric extension is provided in Appendix \ref{sec:sdp}.

\begin{defn}[Separable outer approximation for DC processes] \label{def:out}
A bipartite ordered process $W$ belongs to $\mathcal{L}_\text{DC}^{\text{out},\text{PPT}_k}$ if $W$ has a PPT k-symmetric extension on the bipartition $A_I|A_O B_I$.
\end{defn}

One could be tempted to assume that all process which are separable in $A_I|A_OB_I$  bipartition have a direct-cause explanation. However, in Section\,\ref{sec:separableDC}, we prove that this is not the case by constructing an explicit example. This ensures that the problem of certifying if a given process is not direct-cause is not equivalent to certifying entanglement.


\subsection{Common-cause processes}

Common-cause processes are those where the correlations are not due to communication from Alice to Bob, but only due to a bipartite quantum state $\rho\in\mathcal{L}(A_I\otimes B_I)$ initially shared by them. In this way, a common-cause process does not involve any particular quantum dynamic and can be fully characterized by a fixed bipartite state. In order to establish a notation which dialogues to general scenarios which may have non-trivial dynamics between Alice and Bob, we describe common-cause processes by $\Wcc = \rho^{A_I B_I} \otimes \identity^{A_O}$. 

\begin{defn}[Common-cause process]
\label{defn:ccproc}
	A linear operator $W_{\text{CC}}\in \mathcal{L}(A_I \otimes A_O \otimes B_I)$ is a common-cause process if it can be written as
\begin{equation}
\label{eq:ccprocesses}
\Wcc = \rho^{A_I B_I} \otimes \identity^{A_O}, 
\end{equation}
where $\rho^{A_I B_I}$ is a quantum state shared between Alice and Bob, \textit{i.e.,} $\rho^{A_I B_I}\succeq 0$, $\tr(\rho^{A_I B_I})=1$, and $\identity^{A_O}$ is the identity operator acting on $A_O$. 
We denote the set of common-cause processes by $\mathcal{L}_{\textrm{CC}}$.  
\end{defn}

\subsection{Classical CCDC processes}

Classical CCDC processes are processes which can be decomposed as a convex combination of a common-cause and a direct-cause ones. Such processes can always be understood as simple classical statistical mixtures between quantum common-cause and quantum direct-cause processes \cite{Ried2015,Feix2016}. 

\begin{defn}[Classical CCDC processes]
A linear operator $W_{\text{CCDC}}\in \mathcal{L}(A_I \otimes A_O \otimes B_I)$ is a classical CCDC process if it can be written as
    \begin{equation}
        W_{\textrm{CCDC}} := p \Wcc + (1-p) \Wdc, 
    \end{equation}
    where $\Wcc \in \mathcal{L}_{\textrm{CC}}$, $\Wdc \in \mathcal{L}_{\textrm{DC}}$, and $p \in [0,1]$. We denote the set of classical CCDC processes by $\mathcal{L}_{\textrm{CCDC}}$, which is the convex hull of $\mathcal{L}_{\textrm{CC}} \cup \mathcal{L}_{\textrm{DC}}$.
\end{defn}

\subsection{Bipartite ordered processes}
\label{subse:supermaps}
	
	Previously, we have only considered processes which do not require any explicit use of auxiliary systems or environments. In a more general process, the initial state $\rho$ received by Alice may be defined in a larger space, and Alice's lab can only operate on part of it. In order to tackle this situation, we now consider that the initial state is defined not only on Alice's input system, but also in some auxiliary system with unconstrained finite dimension, that is, $\rho\in \mathcal{L}(A_I\otimes \aux)$. Now, Alice can only operate on the part of $\rho$ that enters her laboratory, that is, she can perform any operation $\map{\Lambda}:\mathcal{L}(A_I)\to \mathcal{L}(A_O)$, which will transform the initial state $\rho$ into%
\footnote{Here $\map{\identity}^{\text{aux}}$ stands for the identity map on the auxiliary space, that is, for any operator $A\in\mathcal{L}(\aux)$ we have $\map{\identity}^{\text{aux}}(A)=A$.}%
$\map{\Lambda}^{A_I\to A_O}\otimes \map{\identity}^{\text{aux}}(\rho)$,
	 which will be subjected to a global dynamics described by a quantum channel $\map{D}:\mathcal{L}(A_O\otimes\aux)\to \mathcal{L}(B_I)$.  See Fig.\,\ref{fig:oneslotcombequiv} for a circuit based pictorial illustration.
	 
	 \begin{figure} 
    \centering
    \includegraphics[width=8.8cm]{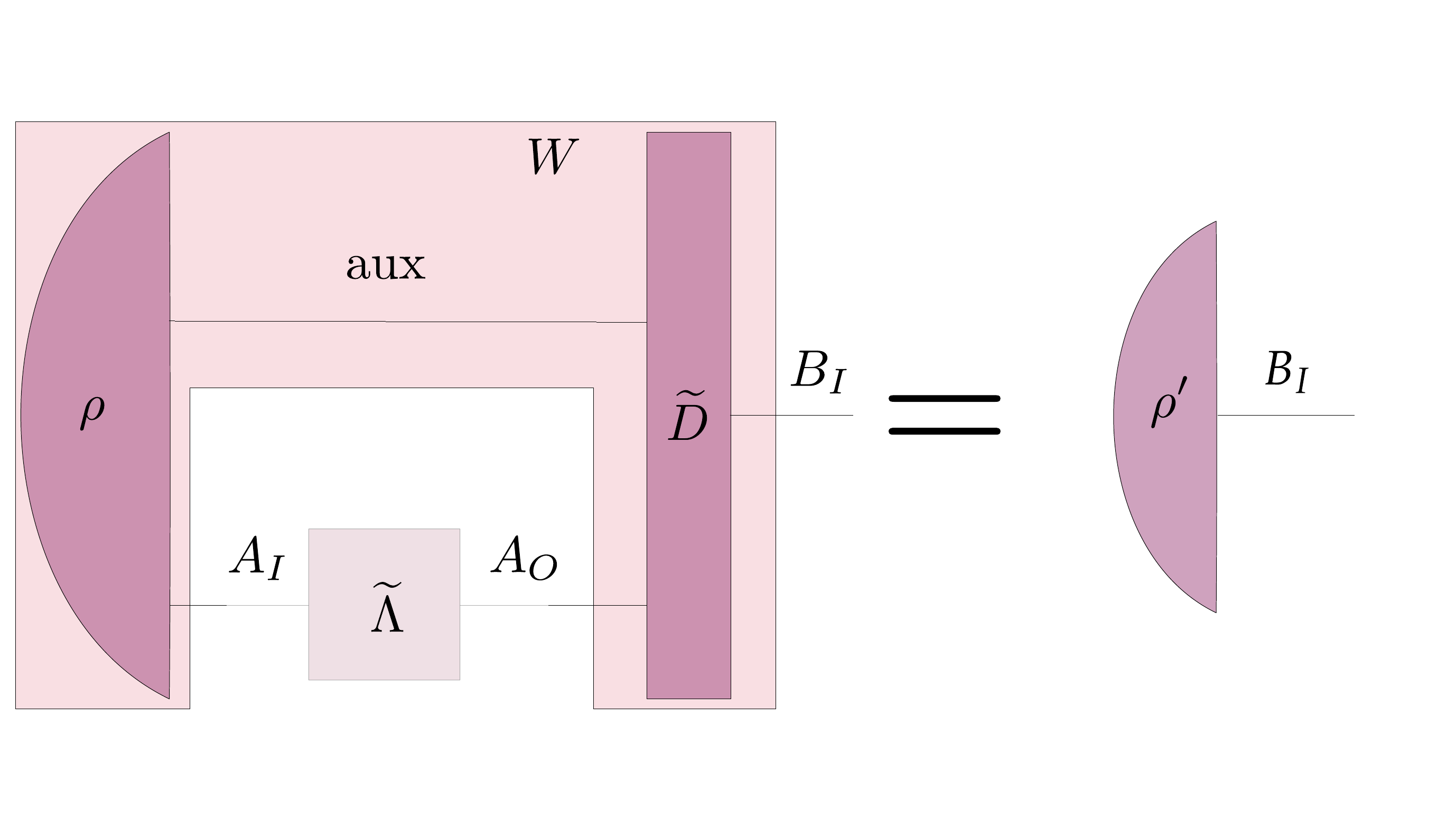}
    \caption{A circuit illustrating a bipartite ordered process $W=\rho*D$ where $\rho \in \mathcal{L}(A_I \otimes \aux)$ is a quantum state and $D$ is the Choi operator of a channel $\widetilde{D}: \mathcal{L}(A_O \otimes \aux) \rightarrow \mathcal{L}(B_I)$. When a quantum channel $\widetilde{\Lambda}:\mathcal{L}(A_I) \rightarrow \mathcal{L}(A_O)$ with Choi operator $\Lambda$ is ``plugged'' into the process $W$, the state $\rho'=W*\Lambda \in \mathcal{L}(B_I) $ is obtained.} 
    \label{fig:oneslotcombequiv}
\end{figure} 

	A bipartite ordered process can then be described by the known quantum state
 $\rho\in\mathcal{L}(A_I\otimes \aux)$ and the known dynamics represented by the channel $\map{D}:\mathcal{L}(A_O\otimes \aux)\to \mathcal{L}(B_I)$. In the Choi-Jamio\l{}kowski representation, this can be conveniently described by $W := \rho^{A_I \aux} * D^{A_O \aux/ B_I}$. Under this definition, for any quantum operation $\map{\Lambda}:\mathcal{L}(A_I)\to \mathcal{L}(A_O)$
	performed by Alice, the quantum state arriving in Bob's input space can be obtained via the link product as
	\small
	\begin{align}
	\begin{split}
	 W * \Lambda^{A_I/A_O} &=\rho^{A_I \aux} * D^{A_O \aux/ B_I} * \Lambda^{A_I/A_O} \\
		&= D^{A_O/B_I}*\Lambda^{A_I/A_O} *\rho^{A_I \aux} \\
		&=D^{A_O/B_I} * \map{\Lambda}^{A_I\to A_O}\otimes \map{\identity}^{\text{aux}}\big(\rho^{A_I \aux}\big)  \\
		&= \map{D}\Big(\map{\Lambda}^{A_I\to A_O}\otimes \map{\identity}^{\text{aux}}\big(\rho^{A_I \aux}\big)\Big).
	\end{split}
\end{align}	 
\normalsize
\begin{defn}[Bipartite ordered process]
    A linear operator $W \in \mathcal{L}(A_I \otimes A_O \otimes B_I)$ is a bipartite ordered process if there exists a quantum state $\rho^{A_I \aux} \in \mathcal{L}(A_I \otimes \aux)$,\textit{i.e.,} $\rho^{A_I \aux} \succeq 0$, $\tr(\rho^{A_I \aux})=1$ and a quantum channel $\map{D}: \mathcal{L}(A_I \otimes \aux) \rightarrow \mathcal{L}(B_I)$ with Choi operator $ D^{A_O \aux/ B_I}$, \textit{i.e.,} $D^{A_O \aux/ B_I}\succeq0$, $\tr_{B_O}(D^{A_O \aux/ B_I})=\identity_{A_0,\aux}$, such that
\small   
    \begin{align}
    	\begin{split}
        W :=& \rho^{A_I \aux} * D^{A_O \aux/ B_I} \\
           =& \tr_\aux\Bigg( \left(\rho^{A_I \aux}\otimes \identity^{A_OB_I}\right)^{T_{\aux}} \\
           & \hspace{0.7cm}\cdot \left( \identity^{A_I} \otimes D^{A_O \aux/ B_I}\right) \Bigg) .
           \end{split}
    \end{align}
\normalsize
We denote the set of bipartite ordered processes  by $\mathcal{L}_{\AB}$.
\end{defn}

	Bipartite ordered process may be seen as quantum objects transforming a quantum operation into a quantum state%
\footnote{In  Ref.\,\cite{Chiribella2008, Chiribella2009a} the authors prove that under a set of assumptions, it can be proven that bipartite ordered process are indeed the most general method to transform arbitrary quantum operations into quantum states.}%
and are also known in the literature as quantum non-Markovian processes \cite{Pollock2015}, causal maps \cite{Ried2015}, quantum co-strategies \cite{Gutoski2007}, and ordered process matrices \cite{Oreshkov2012}. Also, ordered processes may also be seen as a particular instance of quantum channels with memory \cite{kretschmann05} and quantum combs \cite{Chiribella2009a}.  

	In this general definition of bipartite ordered processes, Markovian processes form the particular case where there is no auxiliary space $\aux$. Common-cause processes can also be represented as bipartite ordered processes. Indeed, if $\Wcc=\rho^{A_I B_I} \otimes \identity^{A_O}$, we can set the auxiliary space $\mathcal{H}_\text{aux}$ as isomorphic to $\mathcal{H}_{B_I}$ and the initial state $\rho^{A_I\text{aux}}$ as isomorphic to $\rho^{A_I B_I}$. Then, one can set the channel $\map{D}$ as an operation that sends the auxiliary system to $B_I$ and discard the system in $A_O$. This can be done formally via 
\small
\begin{align}
\label{eq:ccprocesslink}
\Wcc &= \tr_{\aux'}\left[\rho^{A_I \aux} * \keketbra{\uswap}{\uswap}^{A_O \aux /B_I \aux' }\right] \nonumber \\
	        &= \rho^{A_I B_I} \otimes \identity^{A_O}, 
\end{align}
\normalsize
where $\textrm{U}_{\textrm{SWAP}}:= \sum_{ij} \ketbra{ij}{ji}$ is the swap operator. A circuit illustration of common-cause process as bipartite ordered ones is presented in Fig.\,\ref{fig:circuitcc}.
\begin{figure}
    \centering
    \includegraphics[width=10cm]{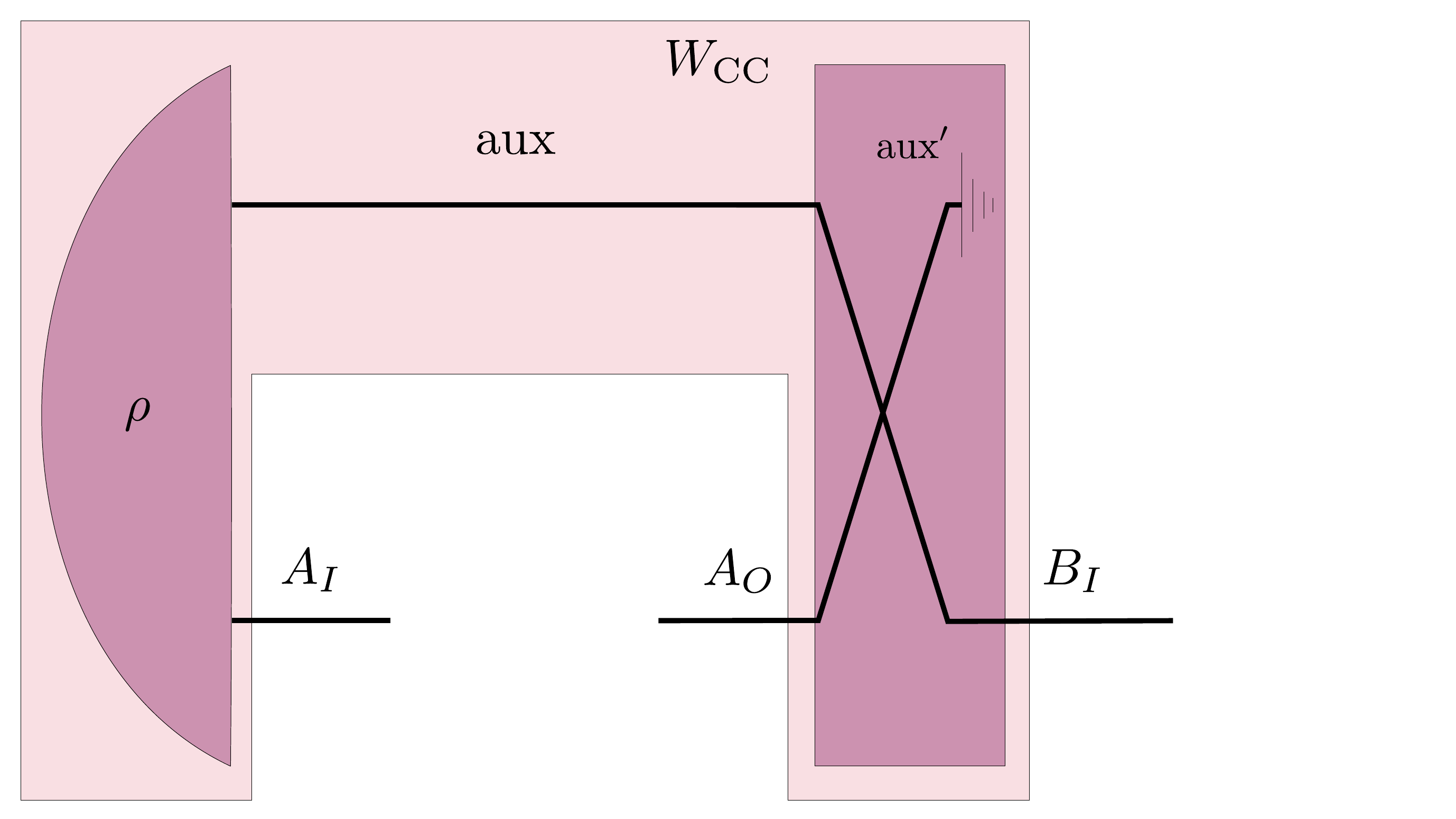}
    \caption{A general bipartite common-cause process can be represented by a circuit of this form. The shared quantum state between Alice and Bob comes from an initial correlated state between Alice and an auxiliary system. The state of the auxiliary system is exchanged with the system coming from the output of Alice and sent to Bob, and the auxiliary system is then discarded, which makes the output state of Alice to have no impact on the input of Bob. This makes only the input states of Alice and Bob to be relevant for the process, which is why the correlations between them stem from the common-cause correlations on $\rho$.}
    \label{fig:circuitcc}
\end{figure}

With all the important processes defined, Fig.\,\ref{fig:ccdcsets} shows the relationship between all the sets defined in this section. 

\begin{figure}[t]
    \centering
    \includegraphics[width=7cm]{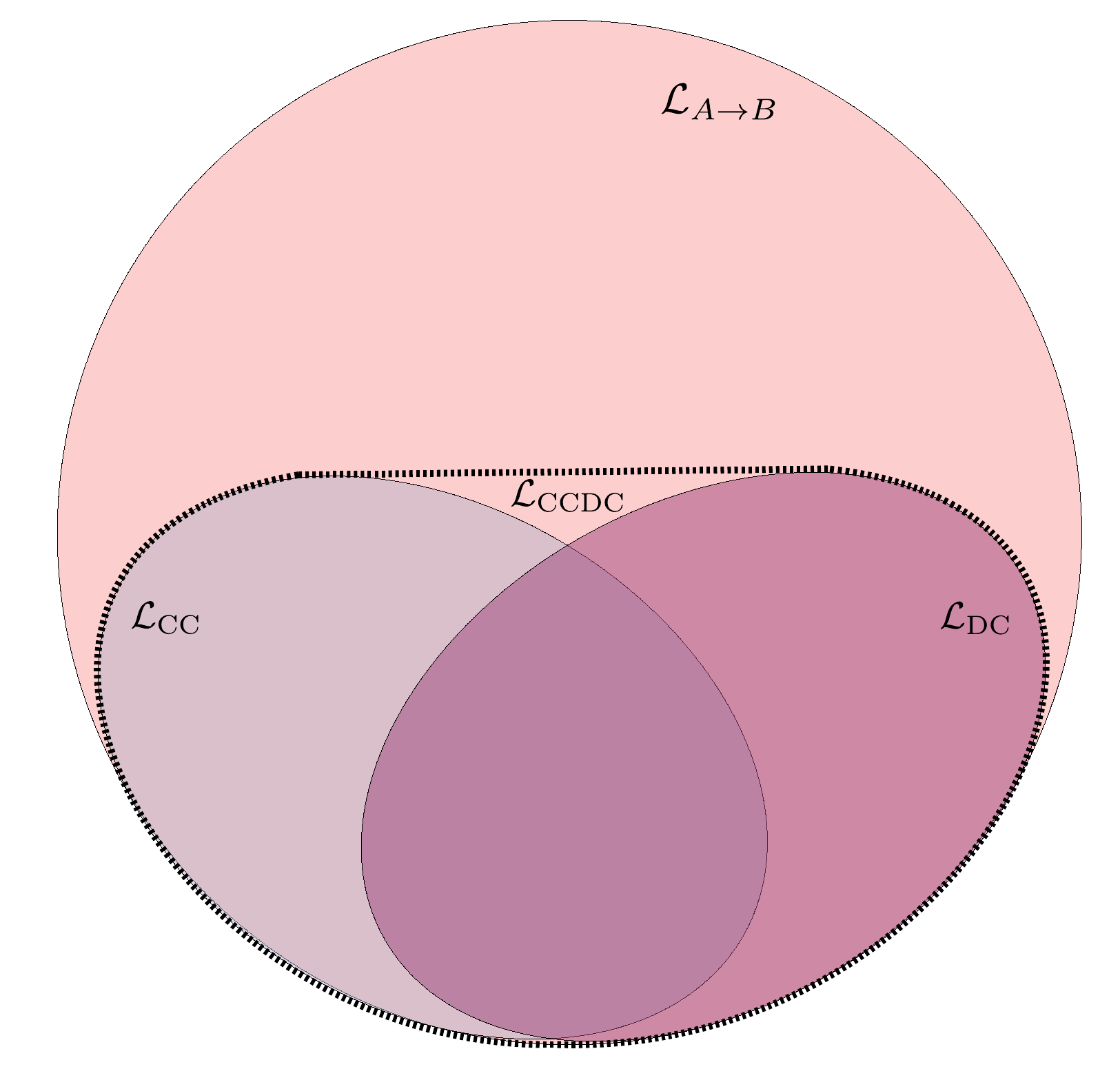}
    \caption{Illustrative representation of the sets in the CCDC scenario. The sets of common-cause processes ($\mathcal{L}_{\textrm{CC}}$) and of direct-cause processes ($\mathcal{L}_{\textrm{DC}}$) are proper subsets of the set of bipartite ordered process $\mathcal{L}_{\AB}$, with a non-empty intersection between $\mathcal{L}_{\textrm{CC}}$ and $\mathcal{L}_{\textrm{DC}}$. Classical CCDC processes are all processes in the convex hull of $\mathcal{L}_{\textrm{CC}} \cup \mathcal{L}_{\textrm{DC}}$, represented by $\mathcal{L}_{\textrm{CCDC}}$.}
    \label{fig:ccdcsets}
\end{figure}	

In addition to its definition, bipartite ordered processes admit a characterization in terms of linear and positive semi-definite constraints, which will be useful later in this work.
If follows from direct inspection that all bipartite ordered process $W \in \mathcal{L}(A_I \otimes A_O \otimes B_I)$  respect
\begin{subequations}
\label{eq:condordproc}
    \begin{align}
        W & \succeq 0, \\
        \tr_{B_I}(W) & = \sigma^{A_I}\otimes \identity^{A_O}, \label{subeq:condprojab} 
    \end{align}
\end{subequations}
where $\sigma^{A_I}$ is a quantum state. Interestingly, this condition is also sufficient \cite{Chiribella2009a}. For any linear operator $W$ respecting the conditions in Eq.\,\eqref{eq:condordproc}, one can find a quantum state $\rho\in \mathcal{L}(A_I\otimes \aux) $ and a quantum channel  $\map{D}:\mathcal{L}(A_O\otimes \aux)\to \mathcal{L}(B_I)$ such that $\rho^{A_I\text{aux}}*D^{\text{aux}B_I/B_O}=W$. One possible construction is done by setting the auxiliary space $\aux$ to be isomorphic to $A_I$ and $\rho\in \mathcal{L}(A_I\otimes \aux) $ be be isomorphic to a purification of $\sigma^{A_I}$, for instance, 
\small
\begin{equation}
\rho:= \left(\identity^{A_I}\otimes\sqrt{\sigma^{\aux}}^T\right) \; \ketbra{\identity}{\identity}^{A_I\text{aux}}\; \left(\identity^{A_I}\otimes\sqrt{\sigma^{\aux}}^T\right).
\end{equation}
\normalsize
 Now, one can define a quantum channel via 
 \small
\begin{align}
	\begin{split}
D:=&  \left(\sqrt{\sigma^{\aux}}^{-1} \otimes \identity^{A_OB_I} \right) \\
   & W^{\text{aux}A_OB_I} \; \left(\sqrt{\sigma^{\aux}}^{-1} \otimes \identity^{A_O B_I}\right),
   \end{split}
\end{align}
\normalsize
where $\sqrt{\sigma}$ is the unique positive semi-definite square root of $\sigma$ and $\sigma^{-1}$ is the Moore–Penrose inverse of $\sigma$, that is, the inverse of $\sigma$ on its range. In this way, direct inspection shows that $D$ is the Choi operator of a quantum channel and that $W=\rho*D$.

Inspired by Ref.\,\cite{Araujo2015}, we now present another characterization for bipartite ordered processes which will be useful for proving our results. A linear operator $W \in \mathcal{L}(A_I \otimes A_O \otimes B_I)$ is a bipartite ordered process if and only if it respects
\begin{subequations}
\label{eq:ordproc}
\begin{align}
W & \succeq 0, \\
W & = L_{\AB}(W), \\
\tr(W) & = d_{A_O} \label{subeq:condnormab},
\end{align}
\end{subequations}
where
\begin{equation}
    \label{eq:projab}
    L_{\AB}(W) := W + \tracerep{A_O B_I}{W} - \tracerep{B_I}{W},
\end{equation}
with $\tracerep{X}{(\cdot)} = \tr_X (\cdot) \otimes \frac{\identity^{X}}{d_X}$ and $L_{\AB}(W)$ is the map projecting an operator $W$ into the linear space spanned by bipartite ordered processes. Eqs.\,\eqref{eq:condordproc} and \eqref{eq:ordproc} are equivalent conditions for an operator $W \in \mathcal{L}(A_I \otimes A_O \otimes B_I)$ to be a valid bipartite ordered process. In particular, conditions given by Eqs.\,\eqref{eq:ordproc} are rather useful for implementing our numerical methods, due to the use of the projector $L_{\AB}(W)$.

\section{Detecting and quantifying non-classical CCDC}
\label{sec:cetifnccdc}

Inspired by the robustness of entanglement \cite{Vidal1999,Steiner2003}, one can quantify the violation of a classical CCDC decomposition in a process $W \in \mathcal{L}(A_I \otimes  A_O \otimes B_I)$ in terms of its \textit{generalized robustness of non-classical CCDC}\footnote{Note that the definition of generalized robustness $R_G$ here has a one-to-one relation to the ``non-classicality of causality'' $\mathcal{C}$ presented in Ref.\cite{Feix2016} via $R_G(W)=\frac{\mathcal{C}(W)}{1+\mathcal{C}(W)}$. } \cite{Feix2016} by

\small
\begin{subequations}
\label{eq:genrob}
    \begin{align}
            R_{G} (W) := \textrm{min}  \hspace{.2cm}  & r \nonumber \\
             \textrm{subject to}   \hspace{.2cm} &  (1-r)W + r \Omega \in \mathcal{L}_{\textrm{CCDC}}, \label{subeq:genrobclas}\\
                             &  0 \leq r \leq 1, \\        
                              & \Omega \in \mathcal{L}_{\AB}, \label{subeq:ordsetcar}
    \end{align}
\end{subequations}
\normalsize
where $\mathcal{L}_{\AB}$ is the set of bipartite ordered processes defined in Eqs.\,\eqref{eq:ordproc}. The generalized robustness $R_G(W)$ corresponds to how resistant is the non-classical CCDC property of the process $W$ against its worst possible noise.

In Appendix \ref{app:robustness_bound}, we show that the generalized robustness of all bipartite processes is upper-bounded by $R_G(W) \leq 1 - \frac{1}{d_{A_I}}$, which is attainable with the examples of non-classical CCDC processes we present in Section \ref{sec:ncccdc}.

Another quantification of non-classical CCDC is in terms of its \emph{white-noise robustness}, given by 
\small
\begin{subequations}   
\label{eq:wnrob}
   \begin{align}
   R_{WN} (W) :=  \textrm{min}  \hspace{.2cm}  & r \nonumber \\
             \textrm{subject to}   \hspace{.2cm} &   (1-r) W + r \frac{\identity}{d_{A_I} d_{B_I}}  \in \mathcal{L}_{\textrm{CCDC}}, \label{subeq:wnlclas}\\
                             &   0 \leq r \leq 1, 
    \end{align}
\end{subequations}
\normalsize
where $\identity$ is the identity operator in $\mathcal{L}(A_I \otimes A_O \otimes B_I)$. The value $R_{WN}(W)$ corresponds to how resistant is the non-classical CCDC property of the process $W$ against white noise.
Figure \ref{fig:convcombwit} gives a representation of both types of robustness.

In Appendix \ref{sec:upper_WNR}, we show that the white noise robustness of all bipartite processes is upper-bounded by $R_{WN}(W) \leq 1 - \frac{1}{d_{A_I}d_{A_O}d_{B_I}+1}$. Such bound is not tight, but is useful to demonstrate that the convex optimization problem of Eq.\eqref{eq:wnrob} respects strong duality (see Appendix \ref{sec:strong_duality} for more details).

\begin{figure}
    \centering
    \includegraphics[width=7cm]{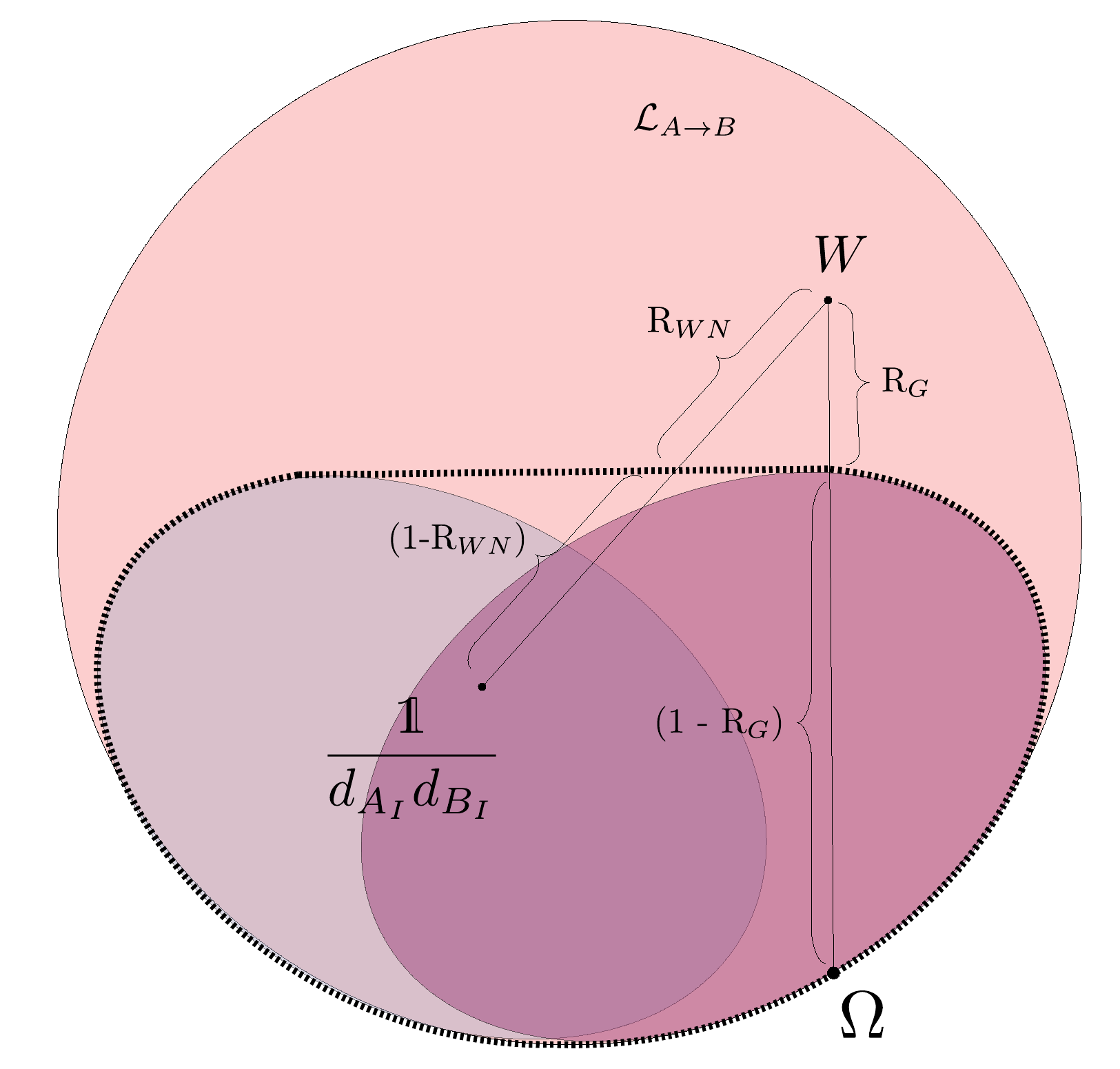}
    \caption{Representation of the problem of determining if a given process $W$ is non-classical CCDC. The above figure illustrates an ordered process operator $\Omega \in \mathcal{L}_{\AB}$ such that the minimum convex combination of it with the process $W$ results in a classical CCDC process.}
    \label{fig:convcombwit}
\end{figure}

	\subsection{Witnessing non-classical CCDC processes}	
	\label{subsec:witness}
	
	A non-classical CCDC witness is a Hermitian operator $S \in \mathcal{L}(A_I \otimes A_O \otimes B_I)$ which respects $\tr(S W_{\text{CCDC}}) \geq 0$ for every classical CCDC process $W_{\text{CCDC}}$ \cite{Feix2016}. Non-classical CCDC witnesses that are not trivial are the ones having $\tr(S W) < 0$ for some non-classical CCDC process $W$, as the main purpose of a witness is to certify that a given process is non-classical CCDC. Since the set of classical CCDC processes is closed and convex, similarly to entanglement \cite{OtfriedReview}, every non-classical CCDC process may be certified by a non-classical CCDC witness. 
Also, as proven in Appendices \ref{sec:strong_duality} and \ref{sec:sdp}, the convex optimization problems used to define generalized and white noise robustness of non-classical CCDC respect a strong duality condition and the dual formulation of our outer approximation can be used to explicitly construct non-classical CCDC witness. 

	Similarly to entanglement witnesses, verifying if an arbitrary operator
$S \in \mathcal{L}(A_I \otimes A_O \otimes B_I)$ is a non-classical CCDC witness is
likely to be computationally hard \cite{Gurvits2003}, as it would require, in principle, to verify that $\tr(SW_{\text{CCDC}}) \geq 0$ for all classical CCDC processes $W_{\text{CCDC}}$. Despite the hardness of the general problem, we now provide simple sufficient (but not necessary) analytical conditions to ensure that $S$ is a non-classical CCDC witness.

	Any operator $S\in \mathcal{L}(A_I \otimes A_O \otimes B_I)$ respecting $\tr_{A_O}(S) \succeq 0$ and $S^{T_{A_I}} \succeq 0$ is a non-classical CCDC witness. This claim follows from the fact that every CC process respects $W_\text{CC}=\tracerep{A_O}{W_\text{CC}}\succeq0$ and every DC process respects $W_\text{DC}^{T{_{A_I}}} \succeq 0$. Hence, 
\footnotesize
	\begin{align} \label{eqs:witnessconds}
		\begin{split}
	\tr(S W_{\text{CCDC}}) & = p~\tr(S W_{\text{CC}}) + (1-p) \tr(S W_{\text{DC}}) \\
	 					   & = p~\tr( S \tracerep{A_O}{W_{\text{CC}}})	+ (1-p) \tr(S^{T_{A_I}} W^{T_{A_I}}_{\text{DC}}) \\
	 					   & = p ~\tr(\tracerep{A_O}{S}   W_{\text{CC}}) + (1-p) \tr(S^{T_{A_I}} W^{T_{A_I}}_{\text{DC}}) \\
	 					    & \geq0 .
		\end{split}
	\end{align}
\normalsize

\section{Beyond the entanglement approximation: SDP hierarchies for \textit{tight} non-classical CCDC certification}

	As stated in Sec.\,\ref{subsubsec:entapprox}, the set of separable processes in the bipartition $A_I|A_OB_I$ provides an outer approximation for the set of DC processes. As pointed in Ref.\,\cite{Feix2016, Giarmatzi2018}, this provides a SDP approach that gives lower bounds for the non-DC and non-classical CCDC robustness of a bipartite ordered process. Indeed, since the relation
$\mathcal{L}_\text{DC}\subseteq \mathcal{L}_\text{DC}^{\text{PPT}_k}$ holds for every natural number $k$, we can define the convex hull $\mathcal{L}_{\text{CCDC}}^{\text{out},\text{PPT}_k}:=\text{conv}(\mathcal{L}_\text{CC},\mathcal{L}_\text{DC}^{\text{out},\text{PPT}_k})$, and the quantity
\begin{subequations}
\label{eq:genrobout}
    \begin{align}
          R_G^{\text{low},\text{PPT}_k}(W) := \textrm{min}  \hspace{.2cm}  & r \nonumber \\
             \textrm{subject to}   \hspace{.2cm} &  (1-r)W + r \Omega \in
\mathcal{L}_{\text{CCDC}}^{\text{out},\text{PPT}_k}            \\
                             &  0 \leq r \leq 1, \\        
                              & \Omega \in \mathcal{L}_{\AB}, 
   \end{align}
\end{subequations}
which gives a lower-bound for the actual generalized robustness, since the relation $R_G^{\text{low,PPT}_k}(W) \leq  R_G(W)$ holds for every $W$.  Analogously, we can also obtain a lower-bound for the white noise robustness by defining $R_{WN}^{\text{low,PPT}_k}(W)$.

	As we show in Sec.\,\ref{sec:separableDC}, approximating the set of DC processes by the set of separable process is rather limited. In particular, there exist processes which are separable in the bipartition $A_I|A_OB_I$ but are non-classical CCDC. Hence, such processes would never be detected by a method solely based in quantum entanglement techniques. In this section, we overcome this problem by presenting two novel SDP hierarchies of sets which converge to the set of CCDC processes. Also, one hierarchy is based in an inner approximation, whereas the other constitutes an outer approximation. Hence, we can obtain a sequence of converging upper and lower bounds on the non-classical CCDC robustnesses. All the sets discussed in this section are illustrated in Fig.~\ref{fig:inoutapprox}.

\begin{figure}[H]
\centering
\includegraphics[width=0.4\textwidth]{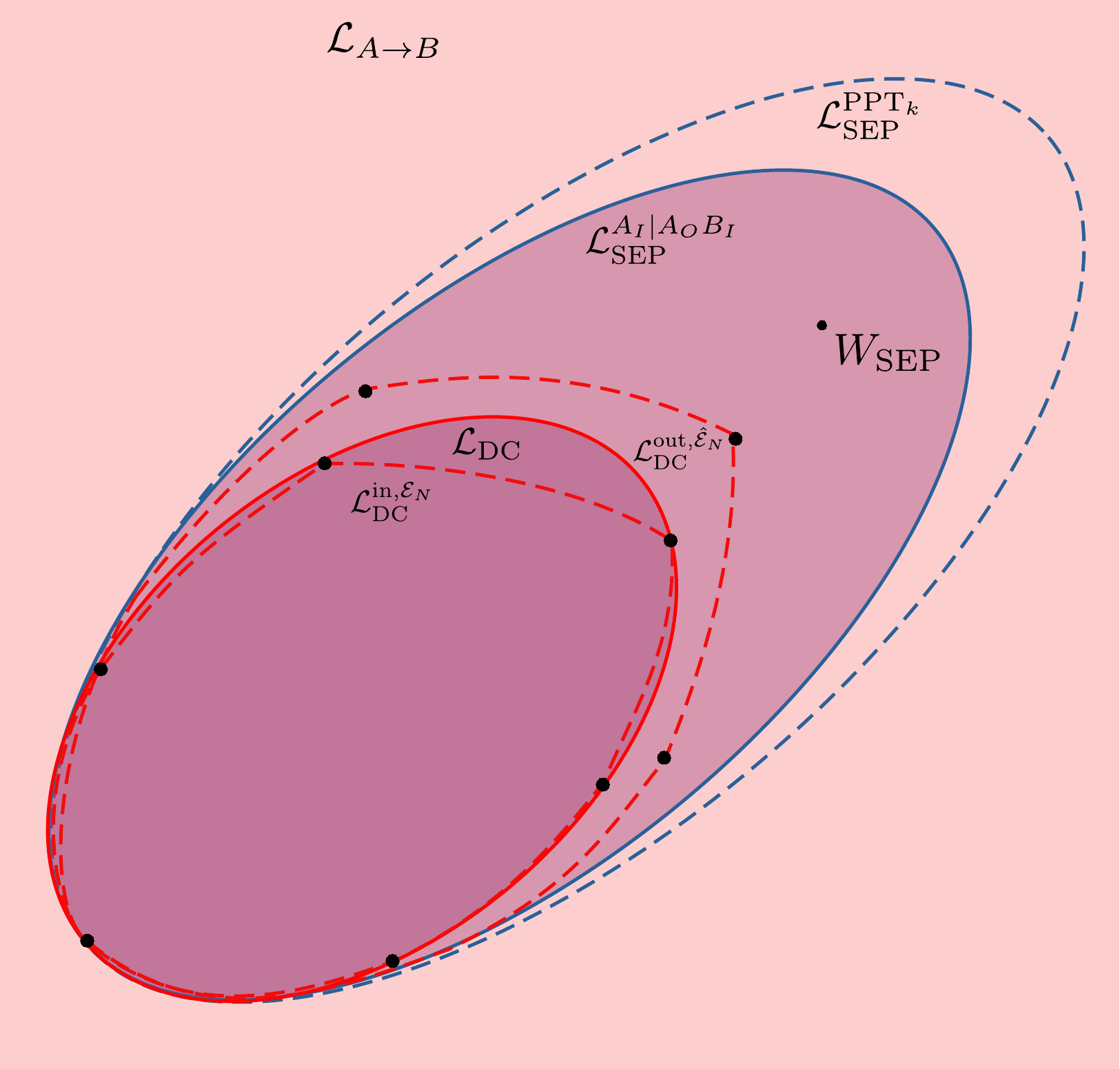}
\caption{\cyan{Hierarchical relation between the set of processes separable in the bipartition $A_I| A_O B_I$ ($\mathcal{L}_{\text{SEP}}^{A_I|A_OB_I}$) and the set of direct-cause processes $\mathcal{L}_{\text{DC}}$. In the limit of large $k$, $\mathcal{L}_{\text{SEP}}^{\text{PPT}_k}$ converges to $\mathcal{L}_{\text{SEP}}^{A_I|A_OB_I}$. By increasing the set of $N$ states, the inner approximation ($\mathcal{L}_{\text{DC}}^{\text{in},\mathcal{E}_N}$) converges to $\mathcal{L}_{\text{DC}}$. Similarly, by increasing the set of $N$ trace-one operators $\{\hat{\rho}_{i=1}^N\}_i$, the outer approximation $\mathcal{L}_{\text{DC}}^{\text{out},\hat{\mathcal{E}}_N}$ also converges to $\mathcal{L}_{\text{DC}}$. The point $W_{\text{SEP}}$ stands for the example of process which is separable in the bipartition $A_I|A_O B_I$, but is non-classical CCDC, defined and discussed in Sec.~\ref{sec:separableDC}.}}
\label{fig:inoutapprox}
\end{figure}

\subsection{Inner approximation for direct-cause processes}
\label{subsec:inndc}

The following inner approximation leads to a hierarchy that converges to the set of direct-cause processes and is inspired by the SDP approach for entanglement detection introduced in \cite{Brandao2004b}.

\begin{defn}[Inner approximation for the direct-cause set] \label{def:inner}
	Let $\mathcal{E}_N:=\{\ket{\psi_i}^{A_I}\}_{i=1}^N$ be a fixed set of pure quantum states $\ket{\psi_i}^{A_I}$. The set $\mathcal{L}_\text{DC}^{\text{in},\mathcal{E}_N}$ corresponds to the inner approximation of the set $\mathcal{L}_{\text{DC}}$ and is composed by every process $W$ that can be written as
\begin{subequations}
\begin{equation}
	W = \sum_i p_i	\ketbra{\psi_i}{\psi_i}^{A_I} \otimes D_i^{A_O/B_I},
\end{equation}
where
\begin{equation}
\sum_ip_i=1 \; \text{ and } \; 	p_i\geq0,\; \forall i,
\end{equation}
and
$D_i^{A_O/B_I}$ are the Choi operators of quantum channels, \textit{i.e.},
\begin{equation}
D_i^{A_O/B_I}\succeq0 \; \text{ and } \; \tr_{B_I}\big(D_i^{A_O/B_I}\big)=\identity^{A_O}, \; \forall i. 
\end{equation}
\end{subequations}
\end{defn}
	Note that in Def.\,\ref{def:inner}, we imposed the partial trace quantum channel normalization condition $\tr_{B_I}\big(D_i^{A_O/B_I}\big)=\identity^{A_O}$, which does not appear in an entanglement-based approach where we only impose positivity and the full trace constraint. As we will see in the further sections, this corresponds to a fundamental difference, since the set of separable quantum states is not equivalent to the set of direct-cause processes.

	By construction, for any set of quantum states $\mathcal{E}_N$, we have the relation $\mathcal{L}_\text{DC}^{\text{in},\mathcal{E}_N}\subseteq \mathcal{L}_\text{DC}$. Also, if $\mathcal{E}_{\infty}$ is the set of all pure quantum states in $\mathcal{L}(A_I)$, we have
$\mathcal{L}_\text{DC}^{\text{in},\mathcal{E}_{\infty}} = \mathcal{L}_\text{DC}$. We can now define the convex hull $\mathcal{L}_{\text{CCDC}}^{\text{in},\mathcal{E}_N}:=\text{conv}(\mathcal{L}_\text{CC},\mathcal{L}_\text{DC}^{\text{in},\mathcal{E}_N} )$, and the quantity
\small
\begin{align}
\begin{split} 
    \label{eq:genrobin}
           R_G^{\text{up,}\mathcal{E}_N}(W) := \textrm{min}  \hspace{.2cm}  & r \\
             \textrm{subject to}   \hspace{.2cm} &  (1-r)W + r \Omega \in \mathcal{L}_{\textrm{CCDC}}^{\text{in,}\mathcal{E}_N}, \\
                             &  0 \leq r \leq 1, \\        
                              & \Omega \in \mathcal{L}_{\AB}, 
\end{split}
\end{align}
\normalsize
which provides an upper-bound for the generalized robustness $R_G(W)$, since the relation $R_G(W) \leq R_G^{\text{up},\mathcal{E}_N}(W)$ holds for every bipartite ordered process $W$ and $\mathcal{E}_N$.

As discussed in details in Appendix \ref{sec:sdp}, the quantity $R_G^{\text{up},\mathcal{E}_N}(W)$ can be evaluated by an SDP. Moreover, this SDP converges to the exact value of $R_G(W)$ when the set of states $\mathcal{E}_N$ contains all pure states in $\mathcal{L}(A_I)$.  Analogously, we can also obtain an upper-bound for the white noise robustness by defining $R_{WN}^{\text{up},\mathcal{E}_N}(W)$.

\subsection{Outer approximation for direct-cause processes}
\label{subsec:outdc}

	For the outer approximation for the set of direct-cause processes, we will make use of an outer approximation for the set of quantum states. Since the set of quantum states is convex, we can always construct an outer approximation given by a polytope, similarly to the strategy used to construct a polytopes providing outer approximations for the set of quantum measurements \cite{quintino15,hirsch16,oszmaniec17,hirsch18,bene18}. Let $\mathcal{S}_d$ be the set of quantum states with dimension $d$, there exists a set of linear operators $\hat{\rho}_i\in\mathcal{L}(\mathbb{C}_d)$ that satisfy $\tr(\hat{\rho}_i)=1$ and such that the convex hull of $\{\hat{\rho}_i\}_i$ contains $\mathcal{S}_d$. Note that the operators $\hat{\rho}_i$ are \emph{not} required to be positive semi-definite, hence they do not correspond to quantum states.

	For the moment, let us assume that we know a set of trace-one operators $\{\hat{\rho}_{i=1}^N\}_i$ such that  $\text{conv}(\{\hat{\rho}_i\}_{i=1}^N) \supseteq \mathcal{S}_d$. We can then use this finite set  to construct an outer approximation for the set of DC processes, similarly to the inner approximation in Def.\,\ref{def:inner}.	
	\begin{defn}[Outer approximation for the direct-cause set] \label{def:outerTRUE}
	Let $\hat{\mathcal{E}}_N:=\{\hat{\rho}_i^{A_I}\}_{i=1}^N$ be a fixed set of linear operators with $\tr(\hat{\rho}_i) = 1 \ \forall i$, such that the convex hull of  $\hat{\mathcal{E}}_N$ contains the set of all quantum states acting on $A_I$. A bipartite ordered process $W$ is in the set $\mathcal{L}_\text{DC}^{\text{out},\hat{\mathcal{E}}_N}$ if it can be written as
\begin{subequations}
\begin{equation}
	W = \sum_i p_i	\hat{\rho}_i^{A_I} \otimes D_i^{A_O/B_I},
\end{equation}
where
\begin{equation}
\sum_ip_i=1 \; \text{ and } \; 	p_i\geq0 \; \forall i,
\end{equation}
and
$D_i^{A_O/B_I}$ are the Choi operators of quantum channels \textit{i.e.}
\begin{equation}
D_i^{A_O/B_I}\succeq0 \; \text{ and } \; \tr_{B_I}\big(D_i^{A_O/B_I}\big)=\identity^{A_O} \; \forall i. 
\end{equation}
\end{subequations}
\end{defn}
By construction, we have that $\mathcal{L}_{\text{DC}} \subseteq \mathcal{L}_\text{DC}^{\text{out},\hat{\mathcal{E}}_N}$ and the set $\mathcal{L}_\text{CCDC}^{\text{out},\hat{\mathcal{E}}_N}:= \text{conv}(\mathcal{L}_{\text{CC}} \cup \mathcal{L}_\text{DC}^{\text{out},\hat{\mathcal{E}}_N})$ is an outer approximation for $\mathcal{L}_{\text{CCDC}}$. The generalized and white noise robustnesses are defined in a way analogous to what is done in Sec.\,\ref{subsec:inndc}, thus being denoted by $R_G^{\text{low},\hat{\mathcal{E}}_N}(W)$ and $R_{WN}^{\text{low},\hat{\mathcal{E}}_N}(W)$, and these are lower bounds for the robustness $R_G(W)$ and $R_{WN}(W)$ respectively.

	We now address the problem of finding sets of trace-one operators which contain the set of quantum states.
Let us first start with the simple qubit case ($d=2$) where the set of quantum states can be faithfully represented in terms of the Bloch sphere \cite{chuang}. For this case, we can construct a set of trace-one operators $\{\hat{\rho}_i\}_i$ such that
	 $\text{conv}(\{\hat{\rho}_i\}_i)\supseteq \mathcal{S}_2$ 
simply by finding a polyhedron that includes a sphere of unit radius. One method to find such polyhedron goes as following:
\begin{enumerate}
\item Sort $N$ normalized vectors $\{\vec{v_i}\}_{i=1}^N$ in $\mathbb{R}^3$. These will be the vertices of a polyhedron inscribed by the Bloch sphere;
\item Find the radius $r_\text{in}<1$ of the largest sphere inscribed by the polyhedron;
\item Construct a polyhedron with vertices given by $\frac{\vec{v_i}}{r_\text{in}}$. This polyhedron includes the Bloch sphere.
\end{enumerate}
	The last step is ensured by the symmetry of the problem, a sphere of radius $r_\text{in}$ is contained in the polytope with norm-one vertices $\{\vec{v_i}\}_{i=1}^N$ if and only if a sphere of radius one contains the polytope with vertices $\frac{\vec{v_i}}{r_\text{in}}$. 
	Also, the radius of the largest sphere inscribed by a polytope (required in step 2) can be made by first finding  the facet representation of this polytope with vertices $\{\vec{v_i}\}_{i=1}^N$, step which can be done via Fourier-Motzkin elimination and can be tackled with the aid of numerical packages such as lrs~\cite{lrs} or the Matlab code \texttt{vert2lcon}~\cite{vert2lcon}. With the facet representation, the radius of the largest inscribed sphere can be found by evaluating the distance of the origin to the hyperplane represented by each facet.
	A concrete example for the qubit case can be found in the Appendix B: ``A family of polyhedra'', of Ref.\,\cite{hirsch17}. In such work, the authors provide a family of polyhedra parameterized by $n\in\mathbb{N}$ where the radius of the larger inscribed sphere is greater than or equals to $r_n:=\cos^2\big(\frac{\pi}{2n}\big)$ and if $n$ is an odd number, the number of vertices is given by $N=2n^2$. This family of polyhedra can be used for a systematic approach to the problem.

	For the general $d>3$, we start by pointing out that every convex set can be approximated by a polytope. Now, in order to obtain a concrete set of operators $\{\hat{\rho}_i\}_i$, we refer to the methods described in Appendix A: ``Calculating shrinking factors'', of Ref.\,\cite{hirsch16} (see also Red.\,\cite{cavalcanti16}).  In a nutshell, the method starts by sampling random pure (which are extremal) states $\ketbra{\psi_i}{\psi_i}$ in the set $\mathcal{S}_d$ and by obtaining the facet representation for the polytope associated to this set. Then, for any fixed ``shrinking factor'' $\eta$, we can verify whether a noisy quantum state $\eta\rho+(1-\eta)\identity/d$ can be written as a convex combination of $\{\ketbra{\psi_i}{\psi_i}\}_i$ by means of an SDP. If all $\eta$-noisy states can be written as convex combinations of $\{\ketbra{\psi_i}{\psi_i} \}_i$, the convex hull of operators is given by $\hat{\rho}_i:=\frac{1}{\eta} \ketbra{\psi_i}{\psi_i} + 1- \frac{1}{\eta} \identity/d$.

	 We observe that we can tighten the outer approximation presented in Def.~\ref{def:outerTRUE} by imposing that the process $W$ should be separable in the bipartition $A_I|A_OB_I$. Also, if we do this by means of the $PPT$ k-symmetric extension (as in Def.~\ref{def:out}), the problem can still be tackled by means of an SDP.
	
\begin{table*}[t]
\centering
\cyan{
{\renewcommand{\arraystretch}{2}
\begin{tabularx}{\textwidth}{|c|c|X|}
\hline
Set         & Robustness  & Requirements for set membership \\ \hline
$\mathcal{L}_{\text{DC}}^{\text{out,PPT}_k}$    & $R^{\text{low,PPT}_k}$ & \footnotesize $ W$ is a valid bipartite ordered process and has a PPT $k$-symmetric extension. \newline (This hierarchy converges to the set of separable processes, which is strictly larger than $\mathcal{L}_\text{DC}$)
\\ \hline
 $\mathcal{L}_{\textrm{DC}}^{\text{out,}\hat{\mathcal{E}}_N}$    & 
 $R^{\text{low},\hat{\mathcal{E}}_N}$ & \footnotesize Given a set of operators $\hat{\mathcal{E}}_N=\{\hat{\rho}^{A_I}_i\}_{i=1}^k$ forming an outer approximation for the set of quantum states, there exist quantum channels $D_i^{A_O/B_I}$ s.t. $W = \sum_i p_i \hat{\rho}^{A_I}_i \otimes D_i^{A_O / B_I}$. \\ \hline 
 $\mathcal{L}_{\textrm{DC}}^{\text{in,}\mathcal{E}_N}$         & $R^{\text{up},\mathcal{E}_N}$ & \footnotesize Given a set of states $\mathcal{E}_N=\{\rho^{A_I}_i\}_{i=1}^k$, there exist quantum channels $D_i^{A_O/B_I}$ s.t. \newline $W = \sum_i p_i {\rho}^{A_I}_i \otimes D_i^{A_O /B_I}$.  \\ \hline
\end{tabularx}
}
}
\caption{
\cyan{
Summary of the sets used as approximation to the set of direct-cause processes, $\mathcal{L}_{\text{DC}}$, ordered from the outside to the inside of the set. The operators $D_i^{A_O/B_I}$ stands for Choi operators of quantum channels, $D_i^{A_O/B_I} \succeq 0$, $\tr_{B_I}(D_i^{A_O/B_I})=\identity$ and $\{p_i\}_i$ is a probability distribution. The operators $\hat{\rho}^{A_I}_i$ are not positive semi-definite, hence they are not quantum states, but the convex hull of the set $\{\hat{\rho}^{A_I}_i\}_{i=1}^k$ is a valid outer approximation for the set of quantum states.  The precise definitions of these sets are given in Def.~\ref{def:out}, Def.~\ref{def:outerTRUE}, and Def.~\ref{def:inner}. These approximations can be used for generating approximations for the set of classical CCDC processes, as $\mathcal{L}_{\text{CCDC}} = \text{conv}(\mathcal{L}_{\text{CC}} \cup \mathcal{L}_\text{DC})$.}
}
\label{table:setsrobssummary}
\end{table*}


\cyan{We present a summary of all the sets defined as approximations to the set $\mathcal{L}_{\text{DC}}$ in Table \ref{table:setsrobssummary}. For obtaining the corresponding inner or outer approximations for the set of classical CCDC processes $\mathcal{L}_{\text{CCDC}}$, it is necessarily to take the convex hull of $\mathcal{L}_\text{CC}$ with the chosen set for approximating $\mathcal{L}_\text{DC}$.} Before finishing this section, we state that, throughout this work, whenever we calculate the non-classical CCDC robustnesses of the inner and outer approximations for $\mathcal{L}_{\text{CCDC}}$, the first lower bounds to be calculated are always provided by the outer approximation $\mathcal{L}_{\text{CCDC}}^{\text{out,PPT}_k}$. If the lower bounds obtained with such approximation does not match the upper bounds, then we take the outer approximation provided by $\mathcal{L}_{\text{CCDC}}^{\text{out,}\hat{\mathcal{E}}_N}$.

\section{Simplest non-classical CCDC processes}
\label{sec:ncccdc}

\subsection{The scenario where $d_{A_I}=d_{A_O}=d$ and $d_{B_I}=d^2$}
\label{subsec:sccdc}

We now present a ``simple'' process which has the highest generalized robustness in the scenario where $d_{A_I}=d_{A_O}=d$ and $d_{B_I}=d^2$. 
In this process, Alice shares a $d$-dimensional maximally entangled state $\ket{\phi^+} = \frac{1}{\sqrt{d}} \sum_{i=0}^{d-1} \ket{ii}$ with the auxiliary system, and the channel of communication to Bob (decoder) is composed of an identity channel, $D = \keketbra{\identity}{\identity}^{\aux A_O/ B_I^1 B_I^2}$ (see Fig.\, \ref{fig:sccdccircuit}).
This process is mathematically represented by
\begin{align}     \label{eq:Wddd2}
	\begin{split}
        W_{ddd^2} & :=  \ketbra{\phi^+}{\phi^{+}}^{A_I \aux} * \keketbra{\identity}{\identity}^{\aux A_O/B_I^1 B_I^2}  \\
                  & = \ketbra{\phi^{+}}{\phi^{+}}^{A_I B_I^1} \otimes \keketbra{\identity}{\identity}^{A_O /B_I^2}.
	\end{split}
\end{align}

\begin{figure}
    \centering
    \includegraphics[width=11cm]{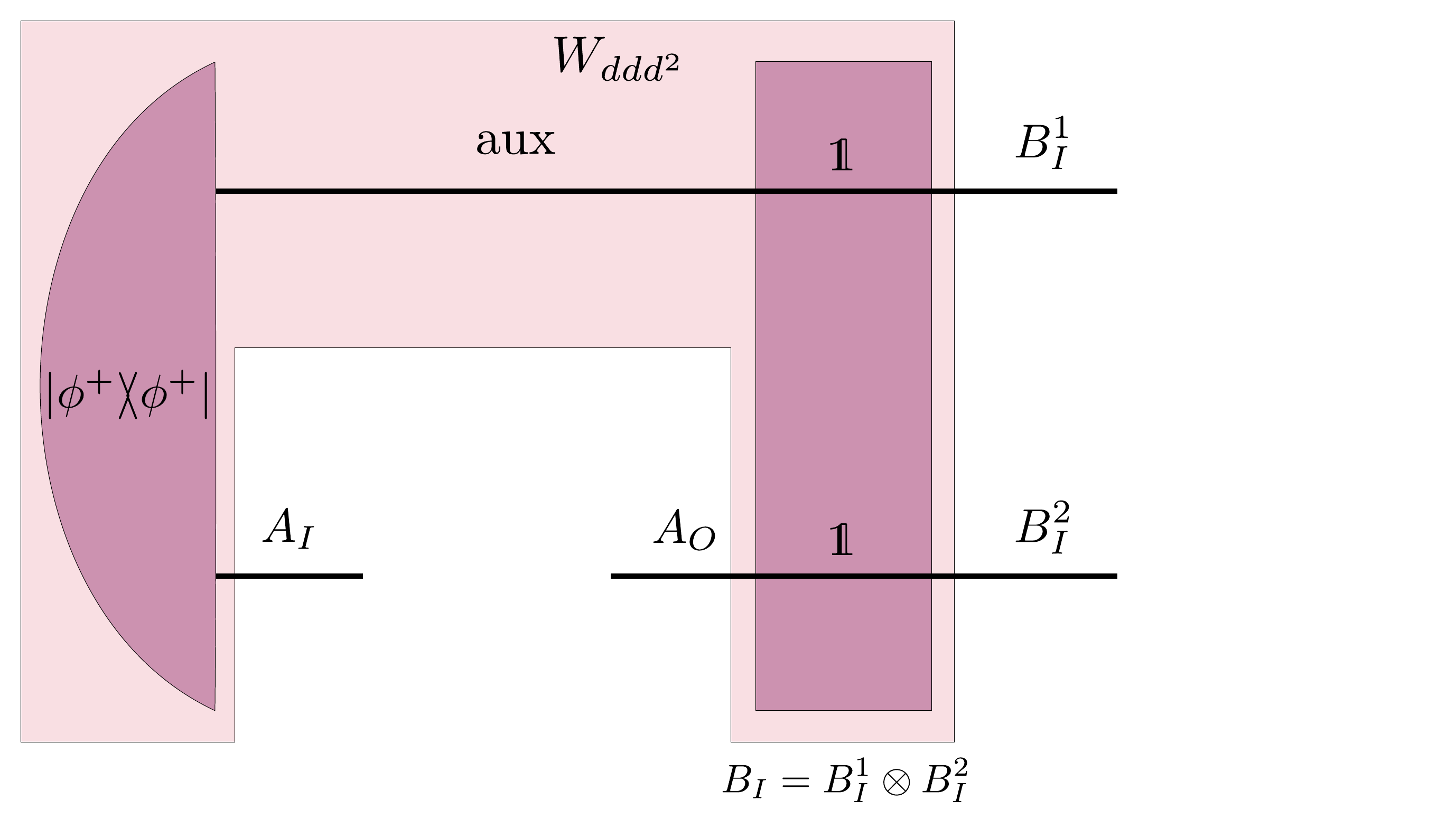}
    \caption{Circuit representation of the process $W_{ddd^2}$. This process consists in the preparation of a maximally entangled state between Alice and an auxiliary system and an identity channel of communication to Bob.}
    \label{fig:sccdccircuit}
\end{figure}

	We argue that this process has a conceptually simple realization, since it only requires the preparation of a maximally entangled state and the identity channel. The process $W_{ddd^2}$ has the interesting property of transforming a quantum channel $\map{\Lambda}: \mathcal{L}(A_I) \rightarrow \mathcal{L}(A_O)$ into a state which is proportional to its Choi operator 
$\Lambda= \sum_{ij} \ketbra{i}{j} \otimes \map{\Lambda}(\ketbra{i}{j})$, that is
\begin{align}
	\begin{split}
	\label{eq:sccdcmaptochoi}
	    W_{ddd^2} * \Lambda &= \left(\map{\identity} \otimes \widetilde{\Lambda} \right)		\ketbra{\phi^+}{\phi^{+}} ,  \\
	                                 &= \frac{\Lambda}{d} .
	\end{split}
\end{align}

Despite the simplicity, we now show that $W_{ddd^2}$ is non-classical CCDC and, moreover, it is the one with highest generalized robustness in the scenario where $d_{A_I} = d_{A_O} = d$ and $d_{B_I}=d^2$.
\begin{thm} \label{thm:lower}
	The bipartite ordered process
	$W_{ddd^2} = \ketbra{\phi^+} {\phi^+}^{A_I B_I^1} \otimes \keketbra{\identity}{\identity}^{A_O /B_I^2}$ 
attains the maximum generalized robustness of all processes with the same dimensions. That is, 
\begin{align}
	\begin{split}
		R_G(W_{ddd^2}) = & \max_{W \in \mathcal{L_{\AB} }} \left[ R_G(W) \right] \\		      
		      = & 1-\frac{1}{d_{A_I}}
	\end{split}
\end{align}
\end{thm}

\begin{proof}

We start the proof by defining the operator 
\begin{align}
	\begin{split}
		S:=& \identity^{A_I A_O B_I} - W_{ddd^2}  \\ 
		 =& \identity^{A_I A_O B_I} -  \ketbra{\phi^+} {\phi^+}^{A_I B_I^1} \otimes \keketbra{\identity}{\identity}^{A_O /B_I^2}
	\end{split}
\end{align}
and showing that $S$ is a valid non-classical CCDC witness, \emph{i.e.}, every CCDC process $W_{\text{CCDC}}$ satisfies $\tr(S W_{\text{CCDC}}) \geq 0$. To ensure that $S$ is a witness, we start by pointing that
\begin{align}
	\begin{split}
\label{eqs:ccwitness}
		\tr_{A_O} S & =  d\identity^{A_I B_I} - \ketbra{\phi^+} {\phi^+}^{A_I B_I^1} \otimes \identity^{B_I^2} \\
		& \succeq 0,
	\end{split}
\end{align}
where the last inequality holds because the smallest eigenvalue of $\ketbra{\phi^+}{\phi^+}^{A_I B_I^1} \otimes \identity^{B_I^2} $ is $1$. 

Now, note that
\begin{align}
	\begin{split}
\label{eqs:dcwitness}
	S^{T_{A_I}} & = \identity - 
	\sum_{ij} \frac{\ketbra{ij}{ji}}{d}^{A_I B_I^1} \otimes \keketbra{\identity}{\identity}^{A_O /B_I^2} \\
	& = \identity -  \frac{\uswap^{A_I B_I^1}}{d} \otimes \keketbra{\identity}{\identity}^{A_O /B_I^2} \\
	& \succeq 0,
	\end{split}
\end{align}
where the last inequality holds because the eigenvalues of $\uswap$ are $+1$ or $-1$. Conditions \eqref{eqs:ccwitness} and \eqref{eqs:dcwitness} together ensure that $S$ is a witness, as shown in Eqs.\,\eqref{eqs:witnessconds}.

Direct calculation shows that $\tr(SW_{ddd^2})=d-d^2=d(1-d)$ and that, for any bipartite ordered process $\Omega$, we have 
\begin{subequations}
\begin{align} 
\tr(S\Omega) &= \tr(\Omega) - \tr(\Omega\,W_{ddd^2}) \\
& \leq \tr(\Omega) = d. \label{ineq:Omega}
\end{align}

Since $S$ is a non-classical CCDC witness, if $(1-r)W_{ddd^2} + r \Omega$ is a CCDC process, it holds that 
${\tr \left( S \left[ (1-r)W_{ddd^2} + r \Omega\right] \right) \geq0}$.
By combining $\tr(SW_{ddd^2})=d(1-d)$ with the inequality \eqref{ineq:Omega}, we have
	\begin{align}
	(1-r) d (1-d) + rd \geq 0, 
	\end{align}
thus $r\geq 1-\frac{1}{d}$.

	We finish the proof by invoking Theorem \ref{theo:upper}, which states that $r \leq 1-\frac{1}{d}$.
	\end{subequations}
\end{proof}

	We now analyze the robustness of $W_{ddd^2}$ against the white noise process. For that, we use the SDP formulation to numerically tackle the case where $d=2$ and $d=3$. For the \cyan{inner approximation,} we have used the set $\mathcal{L}_{\text{CCDC}}^{\text{in}, \mathcal{E}_N}$ with ${\mathcal{E}_{N}}$ being a set with $N=10^4$ uniformly random pure states for $d=2$ and $N = 200$  uniformly random pure states for $d=3$. For the outer approximation, it was enough to use the loose approximation $\mathcal{L}_{\text{CCDC}}^{\text{out,PPT}}$. The results obtained were
\begin{subequations}
	\begin{align}
	&R_{WN}^{\text{up},\mathcal{E}_{10^4}}(W_{222^2}) = R_{WN}^{\textrm{low,PPT}} (W_{222^2}) = 0.8421 \\
	 & R_{WN}^{\text{up},\mathcal{E}_{200}}(W_{333^2})=R_{WN}^{\textrm{low,PPT}} (W_{333^2}) = 0.9529,
	\end{align}
\end{subequations}
where equations hold up to numerical precision. Since the upper and lower-bound coincides, these values are the actual values of robustnesses. 

For $d=2$, we believe that $W_{222^2}$ is the process with maximum white noise robustness on its scenario. Our conjecture is based on a heuristic see-saw technique inspired by \cite{Cavalcanti2017a, Bavaresco2017}, which suggests that the highest value of white noise robustness in the scenario where $d=2$ is $0.8421$. 

In a nutshell, our see-saw technique goes as follows: for a fixed scenario $d_{A_I}=d_{A_O}=d$ and $d_{B_I}=d^2$, we sample a random bipartite ordered process $W$. Then, we obtain its optimal non-classical CCDC witness $S$ from the dual formulation of the PPT white noise robustness (see Appendix \ref{sec:sdp}). We then find the bipartite ordered process which maximally violates the witness $S$, then finding its optimal witness after that. This process is re-iterated until it converges to a robustness value, which we expect to be considerably greater than the robustness of the initially sampled process $W$. The see-saw method is described in details in Appendix \ref{sec:seesaw}.

Moreover, in the scenario where $d=3$, our see-saw method could find a bipartite ordered process $W_\text{max}$ which has $R_{WN}^{\text{up},\mathcal{E}_{200}}(W_\text{max})=R_{WN}^{\text{low,PPT}}(W_\text{max}) = 0.9643 > R_{WN}^{\text{up,}\mathcal{E}_{200}}(W_{333^2})$. This result shows the power of the see-saw method in finding processes which are robust against white noise. It also shows that, even though $W_{ddd^2}$ has maximum generalized robustness for every $d$, which we proved analytically, in the case where $d=3$, $W_{333^2}$ is not the most robust process against white noise. This leads us to conjecture that $W_{ddd^2}$ may be not the most robust process against white noise for $d>3$. 

Before finishing this section we mention that, following the same steps from the demonstration of Theorem \ref{thm:lower}, we can also obtain an analytical lower-bound for white noise robustness of $W_{ddd^2}$, which is
\begin{equation}
R_{WN} (W_{ddd^2}) \geq \frac{d^3}{(d^2 + 1)(d+1)}.
\end{equation}

Differently from the lower-bound presented for generalized robustness, the above inequality is not tight. For instance, when $d=2$, this lower-bound provides $R_{WN}(W_{ddd^2})\geq \frac{8}{15} \approx 0.5333$, which is considerably lower than $R_{WN}^{\text{low,PPT}}(W_{ddd^2}) = 0.8421$. However, it is interesting to point out that the above expression shows that $R_{WN}(W_{ddd^2}) \rightarrow 1$ when $d \rightarrow \infty$.

\subsection{The scenario with minimum dimensions: $d_{A_I}=d_{A_O}=d_{B_I}=2$}
\label{subsec:222}
\begin{figure} 
    \centering
    \includegraphics[width=10cm]{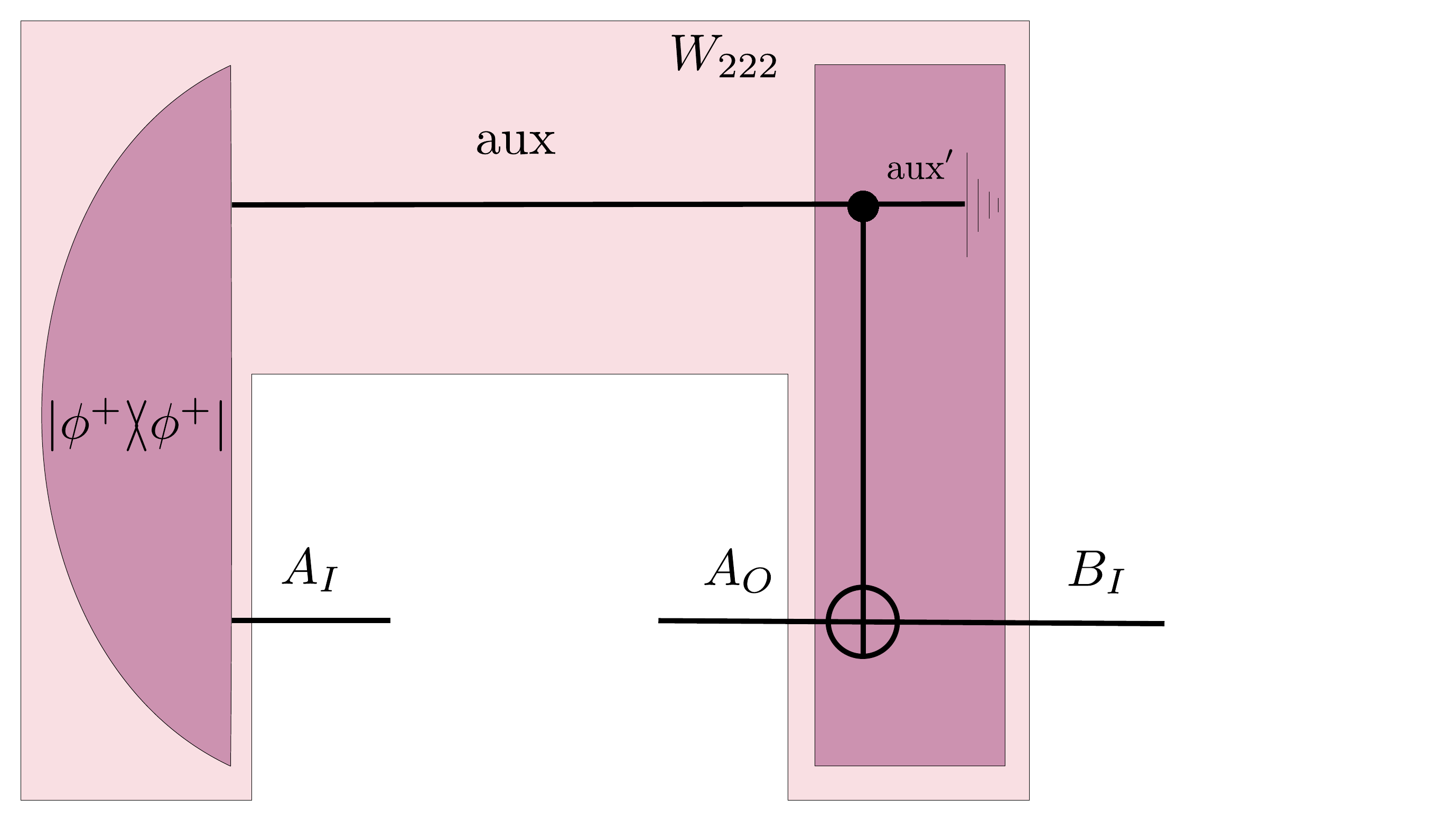}
    \caption{Circuit representation of the process $W_{222}$, constructed with the smallest possible space dimensions $d_{A_I} = d_{A_O} = d_{B_I} = 2$. A maximally entangled state is initially shared between Alice and the auxiliary system. After Alice's operation, a control-NOT is applied, then $\aux'$ is discarded.} 
    \label{fig:cnotproc}
\end{figure}

We now consider the scenario where $d_{A_I}=d_{A_O}=d_{B_I} = 2$, that is, the minimum non-trivial dimensions. 
For this scenario, we propose the process composed by a two-qubit maximally entangled state between Alice and the auxiliary system, and the decoder channel from Alice output to Bob's input being a control-NOT channel, where the auxiliary system is later discarded (see Fig.\, \ref{fig:cnotproc}). Mathematically, such process is described by 
\begin{align} \label{eq:W222}
	\begin{split}
	    W_{222}  :=  \ \tr_{\aux'} \Big( & \ketbra{\phi^+}{\phi^+}^{A_I \aux}  \\
	   *  & \keketbra{\ucnot}{\ucnot}^{\aux A_O/\aux' B_I} \Big) 
	\end{split}
\end{align}
where $\keket{\ucnot}$ is the Choi vector of the control-NOT unitary gate $\textrm{U}_{\textrm{CNOT}}$, given by
\begin{equation}
\textrm{U}_{\textrm{CNOT}} = 
\left( 
\begin{matrix}
1 & 0 & 0 & 0 \\ 
0 & 1 & 0 & 0 \\
0 & 0 & 0 & 1\\
0 & 0 & 1 & 0\\
\end{matrix}
\right).
\end{equation}

Direct calculation shows that the process $W_{222}$ can also be written in terms of a un-normalized $GHZ$ state $\keket{GHZ} := \ket{000} + \ket{111} \in \mathcal{L}(A_I \otimes A_O \otimes B_I)$ :
\small
\begin{equation}
    \label{eq:ghzcomb}
    W_{\textrm{222}}  =  \frac{1}{2} \left(\keketbra{GHZ}{GHZ}  +   \sigma_X^{A_O} \keketbra{GHZ}{GHZ} \sigma_X^{A_O}\right), 
\end{equation}
\normalsize
where $\sigma_X^{A_O}$ is the Pauli matrix $\sigma_X$ applied on subsystem $A_O$. 
The decomposition of Eq.\,\eqref{eq:ghzcomb} expresses $W_{222}$ as a probabilistic mixture of pure ``GHZ-like'' non-normalized states\footnote{Interestingly, one can verify that for every $\epsilon \in \left(0,\frac{1}{2}\right]$,
the operator
$W = \left( \frac{1}{2} + \epsilon \right) \keketbra{GHZ}{GHZ} + \left( \frac{1}{2} - \epsilon \right) \sigma_X^{A_O} \keketbra{GHZ}{GHZ} \sigma_X^{A_O}$ 
is outside $\mathcal{L}_{\AB}$. This ensures that, as illustrated in Fig.\,\ref{fig:convcombghz}, $W_{222}$ is on the boundary of $\mathcal{L}_{\AB}$.} and it will be useful to prove that $W_{222}$ attains the maximum generalized robustness value for its scenario. 

Another decomposition of the process $W_{222}$ is in terms of Pauli matrices:
\small
\begin{align}
	\begin{split}
	    W_{222}  =  \frac{1}{4} &\bigg(\identity^{A_I A_O B_I} + \sigma_Z^{A_I}  \identity^{A_O} \sigma_Z^{B_I}    \\
	               & + \sigma_X^{A_I} \sigma_X^{A_O} \sigma_X^{B_I} - \sigma_Y^{A_I} \sigma_X^{A_O}\sigma_Y^{B_I}\bigg), 
	\end{split}
\end{align}
\normalsize
with implicit tensor product between the operators.

\begin{figure}
    \centering
    \includegraphics[width=8cm]{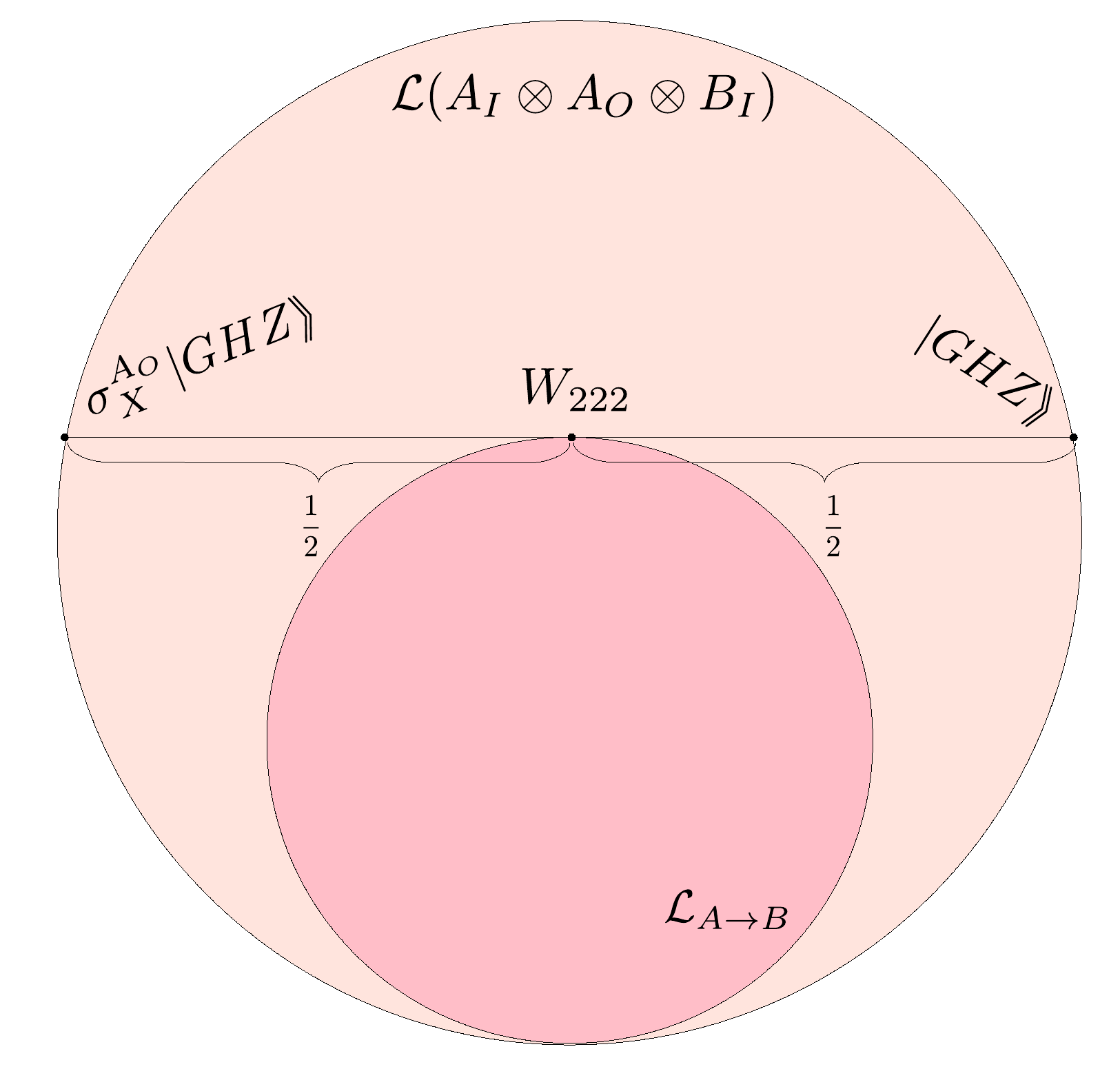}
    \caption{The process $W_{222}$ seen as a convex combination of two GHZ-type un-normalized states acting on $A_I$, $A_O$ and $B_I$. Due to the ordered process causal constraints, GHZ states do not lead to valid bipartite ordered processes, but such convex combination of these two GHZ-type results in the valid bipartite ordered process $W_{222}$.}
    \label{fig:convcombghz}
\end{figure}

\begin{thm} \label{theo:w222}
	The bipartite ordered process
	$W_{222}$ attains the maximum generalized robustness of all processes with the same dimensions. That is, 
\begin{subequations}
	\begin{align}
		R_G(W_{222}) = & \max_{W \in \mathcal{L_{\AB} }}\left[ R_G(W) \right]\\ 
		      		 = & \frac{1}{2}.
	\end{align}
\end{subequations}
\end{thm}

\begin{proof}
The proof of this theorem follows similar steps to the proof of theorem \ref{thm:lower}.
We start by defining the operator 
\small
\begin{subequations}
	\begin{align}
		S:=& \identity - 2W_{222}  \\ 
		 =& \identity -  \left(\keketbra{GHZ}{GHZ}  +   \sigma_x^{A_O} \keketbra{GHZ}{GHZ} \sigma_x^{A_O}\right),
	\end{align}
\end{subequations}
\normalsize
and showing that $S$ is a non-classical CCDC witness. For this, we need to show that $\tr_{A_O}(S) \succeq 0$ and $S^{T_{A_I}} \succeq 0$, as shown in Eqs.\,\eqref{eqs:witnessconds}.
\begin{align}
	\begin{split}
	\frac{1}{2}\tr_{A_O} (S) &= \identity^{A_I B_I} - \left(\ketbra{00}{00} + \ketbra{11}{11} \right) \\
	 &= \ketbra{01}{01} + \ketbra{10}{10} \\
	 &\succeq 0 .
	 \end{split}
\end{align} 

Now, for the second condition, we have
\footnotesize
\begin{align}
	\begin{split}
	S^{T_{A_I}} &= \identity - \left(\keketbra{GHZ}{GHZ}  +   \sigma_x^{A_O} \keketbra{GHZ}{GHZ} \sigma_x^{A_O}\right)^{T_{A_I}} \\
	&=\identity - \sum_{ij}\left( \ketbra{jii}{ijj} + \sigma_x^{A_O} \ketbra{jii}{ijj} \sigma_x^{A_O}\right) \\
	& \succeq 0,
	\end{split}
\end{align}
\normalsize
where the last inequality holds true because $\keketbra{GHZ}{GHZ}^{T_{A_I}}$ and $\left(\sigma_x^{A_O} \keketbra{GHZ}{GHZ}\sigma_x^{A_O}\right)^{T_{A_I}}$ have orthogonal support, \textit{i.e.},
\cyan{ 
\small
\begin{equation}
\keketbra{GHZ}{GHZ}^{T_{A_I}}\left(\sigma_x^{A_O} \keketbra{GHZ}{GHZ}\sigma_x^{A_O}\right)^{T_{A_I}}=0,
\end{equation}
\normalsize
}
\\
and the eigenvalues of $\sum_{ij}  \ketbra{jii}{ijj} $ and $\sum_{ij} \sigma_x^{A_O} \ketbra{jii}{ijj} \sigma_x^{A_O}$ are $+1$ and $-1$.

Direct calculation shows that $\tr(S W_{222})=-2$ and, for every process $\Omega$ in this scenario, we have $\tr(S\Omega)\leq 2$, by an argument analogous to Eq.\,\eqref{ineq:Omega}.

Since $S$ is a non-classical CCDC witness, if $(1-r)W_{222} + r \Omega$ is a CCDC process, it holds that 
${\tr \left(S \left[ (1-r)W_{222} + r \Omega\right] \right) \geq0}$.
It is also true that
\begin{equation}
-2r  + 2(1-r)  \geq 0,
\end{equation}
thus $r\geq\frac{1}{2}$.

	We finish the proof by invoking Theorem \ref{theo:upper}, which states that $r\leq \frac{1}{2}$ in this scenario.
\end{proof}

We also evaluated the robustness of $W_{222}$ against white noise, obtaining evidences that $W_{222}$ attains the maximal white noise robustness on its scenario.

\begin{thm}
	The white noise robustness of the bipartite ordered processes $W_{222}$ is 
	\begin{equation}
		R_{WN}(W_{222}) = \frac{2}{3}.
	\end{equation}
\end{thm}

\begin{proof}
We provide a lower-bound for $R_{WN}(W_{222})$ by using similar steps to the proof of theorem \ref{theo:w222}, starting with the non-classical CCDC witness
\small
\begin{align}
	\begin{split}
		S:=& \identity - 2W_{222}  \\ 
		 =& \identity -  \left(\keketbra{GHZ}{GHZ}  +   \sigma_X^{A_O} \keketbra{GHZ}{GHZ} \sigma_X^{A_O}\right).
	 \end{split}
\end{align}
\normalsize

We have $\tr(S W_{222})=-2$ and $\tr(S)=2^3-4=4$. As $S$ is a non-classical CCDC witness, if $(1-r)W_{222} + r \Omega$ is a CCDC process, ${\tr \left(S \left[ (1-r)W_{222} + r \frac{\identity}{4}\right] \right) \geq0}$. So, it is true that
\begin{align}
  -2(1-r) + r \geq 0,
\end{align}
thus $r\geq\frac{2}{3}$.

We now show, using techniques which are similar to the proof of Theorem \ref{theo:upper}, that the above lower-bound can be attained. First notice that the qubit depolarizing channel $\map{D}_\eta(\rho)=(1-\eta) \rho + \eta \tr(\rho)\frac{\identity}{2}$ is entanglement breaking when $\eta=\frac{2}{3}$ \cite{Horodecki1997a,Horodecki2003}. Also, notice that, since $\tr_{A_I}(W_{222})= \frac{\identity^{A_O B_I}}{d_{B_I}}$, if we apply the depolarizing channel on the subspace $A_I$ of the process $W_{222}$, we obtain
$\map{D}_\eta^{A_I}\otimes \map{\identity}^{A_O B_I}(W_{222})= (1-\eta)W_{222} + \eta \frac{\identity}{4}$.
Since $\map{D}_\eta$ is entanglement breaking for $\eta=\frac{2}{3}$, lemma \ref{lemma:ent_break} ensures that $(1-\frac{2}{3})W_{222} + \frac{2}{3} \frac{\identity}{4}$ is CCDC, thus concluding the proof.
\end{proof}

It is worth to mention that every witness introduced in the proofs of Theorems \ref{thm:lower} and \ref{theo:w222} are the optimal witnesses for these particular processes. 

Similarly to the scenario with $d_{A_I}=d_{A_O}=d$ and $d_{B_I}=d^2$, we implemented the see-saw algorithm that indicates that, in the scenario with $d_{A_I}=d_{A_O}=d_{B_I}=2$, the highest white noise robustness is exactly $\frac{2}{3}$. This suggests that $W_{222}$ has also maximum white noise robustness in its scenario. 

\section{Relation with previous research}
\label{sec:relatwprev}

\subsection{Comparison with the non-classical CCDC process of Ref.\,\cite{MacLean2016}}

We now compare our proposed process with the non-classical CCDC process presented and experimentally implemented in Ref.\,\cite{MacLean2016}. The process consists of the preparation of a maximally entangled state $\ket{\phi^+}$ shared between Alice and the auxiliary system, a partial swap channel from $\mathcal{L}(A_O \otimes \aux)$ to $\mathcal{L}(B_I \otimes\aux')$, which is simply a unitary composed by a coherent mixture of an identity and a SWAP gate, and a partial trace on the output of the auxiliary system. This process can be explicitly written as
\begin{align}
	\begin{split}
	\label{eq:W_MRSR}
	W_{\textrm{MRSR}}  =  \ \tr_{\aux'} & \Big( \ketbra{\phi^+}{\phi^+}^{A_I \aux} \\
	 * & \keketbra{\ups}{\ups}^{A_O \aux/ B_I \aux'}\Big),
	\end{split}
\end{align}
where $\keket{\ups}^{A_O \aux/ B_I \aux'}$ is the Choi vector of the unitary partial SWAP\footnote{
The authors from Ref.\,\cite{MacLean2016} use an equivalent way to represent the unitary partial SWAP gate, which is
\begin{equation} \label{eq:partial_swap}
\ups^{A_I \aux / B_I \aux'} =  \frac{1}{\sqrt{2}} \left( \identity^{A_O \aux / B_I \aux'} + i \identity^{A_O \aux / \aux' B_I} \right). 
\end{equation} In the second term, the identity channel exchanges the outputs $B_I$ and $\aux'$, in comparison to the identity channel in the first term. This part of the channel does exactly the same as $\uswap$ does, differing only by the fact that the output labels are not explicitly exchanged.} gate $\ups$, given by
\begin{equation}
\textrm{U}_{\textrm{PS}} = \frac{1}{\sqrt{2}} \left( \identity + i \, \textrm{U}_{\textrm{SWAP}}\right),
\end{equation}
with $\textrm{U}_{\textrm{SWAP}}$ being the SWAP gate for qubits, given by
\begin{equation}
\label{eq:swap}
\textrm{U}_{\textrm{SWAP}} =	
\left(	
	\begin{matrix}
	1 & 0 & 0 & 0 \\
	0 & 0 & 1 & 0 \\
	0 & 1 & 0 & 0 \\
	0 & 0 & 0 & 1 
	\end{matrix}
\right).
\end{equation}
The label MRSR makes reference to the names MacLean, Ried, Spekkens and Resch, authors of Ref.\,\cite{MacLean2016}. Fig.\,\ref{fig:ried} ilustrates the process $W_{\text{MRSR}}$.

\begin{figure}
    \centering
    \includegraphics[width=8.8cm]{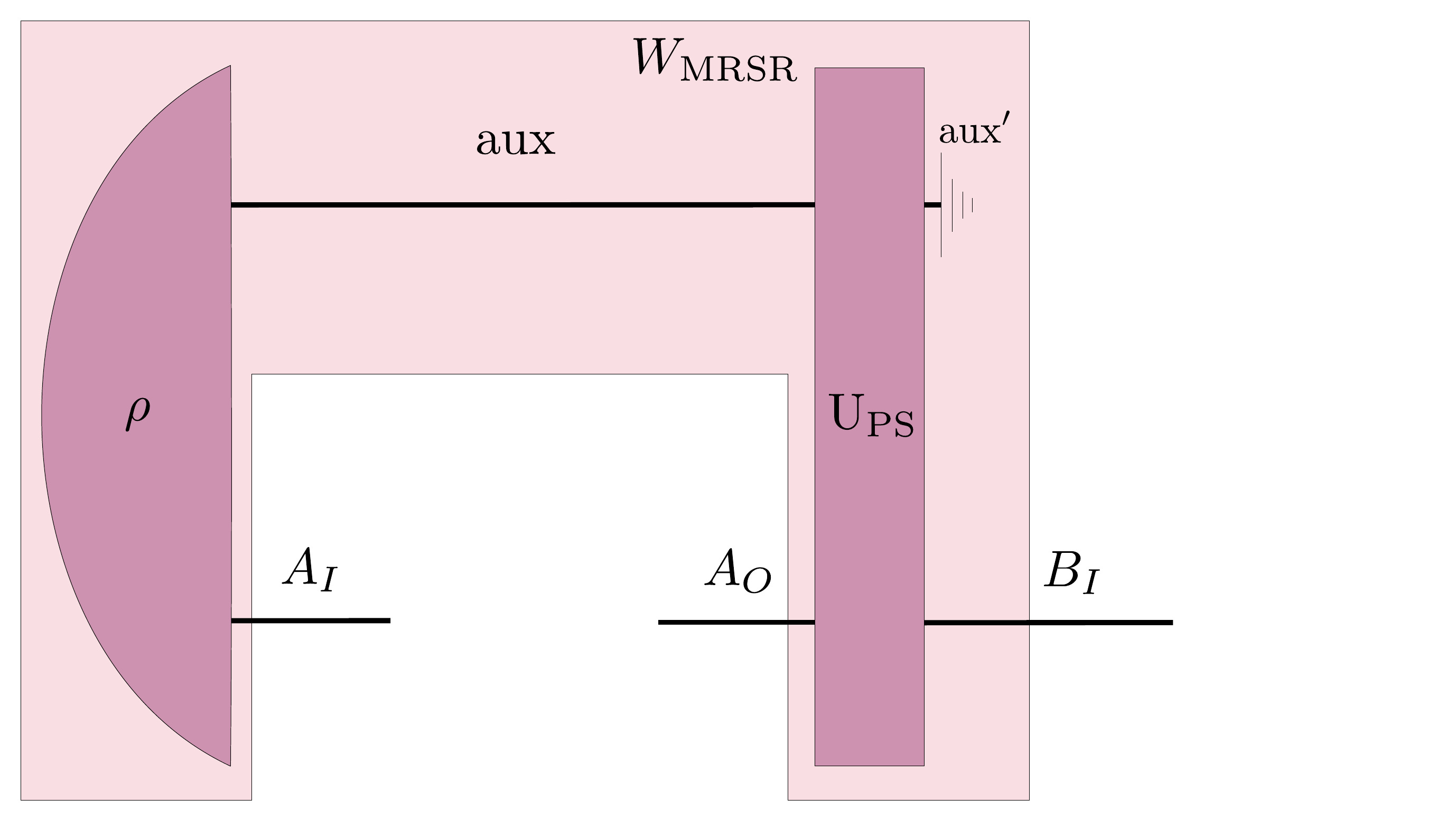}
    \caption{Circuit representation of the process $W_{\text{MRSR}}$ presented in Ref.\,\cite{MacLean2016}. A maximally entangled state is initially shared between Alice and the auxiliary system. After Alice's operation, the partial-SWAP gate $\text{U}_\text{PS}$ is applied, then $\aux'$ is discarded.} \label{fig:ried}
\end{figure}

Using our numerical methods, we can evaluate the values of robustnesses for $W_{\textrm{MRSR}}$, which are
\small
\begin{subequations}
\label{eq:mrsrrobustnesses}
	\begin{align}	
	& R_G^{\textrm{low,PPT}}(W_{\textrm{MRSR}})  = R_G^{\text{up},\mathcal{E}_{10^4}}(W_{\textrm{MRSR}})= 0.3506, \\
	& R_{WN}^{\textrm{low,PPT}}(W_{\textrm{MRSR}}) = R_{WN}^{\text{up},\mathcal{E}_{10^4}}(W_{\textrm{MRSR}}) = 0.5000.
	\end{align}
\end{subequations}
\normalsize

When comparing the robustness values of $W_{\textrm{MRSR}}$ with $W_{\textrm{222}}$ (see Section \ref{subsec:222} and Eqs.\,\eqref{eq:mrsrrobustnesses}), which is a process defined in the same scenario, we verify that $W_{\textrm{222}}$ is strictly more robust against both generalized and white noise than $W_{\text{MRSR}}$. 

\cyan{
For the case of the process $W_{222^2}$, we argue that the construction of $W_{222^2}$ is simpler than $W_{\textrm{MRSR}}$. Both processes require the preparation of a maximally entangled qubit state, but $W_{\textrm{MRSR}}$ requires the implementation of the control swap operation, which is a coherent superposition between the identity channel and the swap channel, while $W_{222^2}$ only requires the identity channel.}

\subsection{Comparison with the non-classical CCDC process of Ref.\,\cite{Feix2016}}
\label{sec:feix}

In Ref.\,\cite{Feix2016}, the authors proposed to study non-classical CCDC by considering quantum superpositions of both relations. Their example is a the tripartite process ordered as $\ABC$. 
A tripartite process ordered as $\ABC$ is an operator $W \in \mathcal{L}(A_I \otimes A_O \otimes B_I \otimes B_O \otimes C_I)$ which can be written as 
\begin{equation}
	W=\rho^{A_I\text{aux}}*D_A^{\text{aux}A_O/B_I\text{aux'}}*D_B^{\text{aux'}B_O/C_I},
\end{equation}
where $D_A$ and $D_B$ are Choi operators of quantum channels.

In a tripartite scenario, common-cause and direct-cause relations can be more complex than in the bipartite scenario. In the case of common-cause relations, more general situations are discussed in Refs.\,\cite{Guo2020, Ringbauer2018}. However, following the definitions presented in Ref.\,\cite{Feix2016}, we restrict the attention to the same common-cause and direct-cause relations of the bipartite case, recovering them when taking $d_{B_O} {=} d_{C_I} {=} 1$. The following definitions are taken considering the ones presented in Ref.\,\cite{Feix2016}.

A tripartite ordered process $W_{\text{CC}}$ is common-cause if
	\begin{equation}
	\tr_{B_O C_I}(W_{\text{CC}}) :=  d_{B_O} \rho^{A_I B_I} \otimes \identity^{A_O},
	\end{equation}
where $\rho^{A_I B_I}$ is a quantum state in $\mathcal{L}(A_I \otimes B_I)$.
A tripartite ordered process $W_{\text{DC}}$ is direct-cause if
	\begin{equation}
	\tr_{B_O C_I} (W_{\text{DC}}) := d_{B_O} \sum_i p_i \rho_i^{A_I} \otimes D_i^{A_O/B_I},
	\end{equation}
where $\rho_i^{A_I}$ are quantum states in $\mathcal{L}(A_I)$ and $D_i^{A_O/B_I}$ are quantum channels from $A_O$ to $B_I$.

A tripartite ordered process $W_{\text{CCDC}} $ is classical CCDC if it can be decomposed in a convex combination of a common-cause process and a direct-cause process. Note that when bipartite processes are considered, \emph{i.e.}, the dimensions of $B_O$ and $C_I$ are equal to one, their definition are equivalent to the ones presented in section \ref{sec:ccdcscen}. 

The example of non-classical CCDC process presented by the authors is the operator\footnote{The state $\ket{\phi^{+}}$ appearing in the expression of $W_{\text{FB}}$ is originally written, in Ref.\,\cite{Feix2016}, as generic bipartite state $\ket{\psi}$. However, we verified that the numerical robustness results obtained in Ref.\,\cite{Feix2016} are reproduced when $\ket{\psi}$ is a maximally entangled state, which lead us to take $\ket{\psi} = \ket{\phi^+}$, as this state is used in every other process mentioned before in this work.}
 $W_{\text{FB}} = \keketbra{W_\text{FB}}{W_{\text{FB}}}$, where 
\begin{align} \label{eq:W_FB}
\begin{split}
\keket{W_{\text{FB}}}  = \frac{1}{\sqrt{2}} \Big( & \ket{\phi^{+}}^{A_I B_I} \keket{\identity}^{A_O/C_I^1} \keket{\identity}^{B_O/C_I^{2}} \ket{0}^{C_I^3}\\
 + & \ket{\phi^{+}}^{A_I C_I^1} \keket{\identity}^{A_O/B_I} \keket{\identity}^{B_O/C_I^{2}} \ket{1}^{C_I^3} \Big),
\end{split}
\end{align}
The label FB makes reference to the names Feix and Brukner, authors of Ref.\,\cite{Feix2016}.

This process can be seen as a superposition of a common-cause and a direct-cause processes, as the first term corresponds to a common-cause process, whereas the second term corresponds to a direct-cause process. We can also represent $W_{\text{FB}}$ in terms of ordered quantum circuits%
\footnote{Indeed, every ordered quantum process can be represented in terms of ordered quantum circuits by concatenating quantum states and quantum operations \cite{Chiribella2009a}.}
 as illustrated in Fig.\,\ref{fig:feix}.
\begin{figure}
    \centering
    \includegraphics[width=8.8cm]{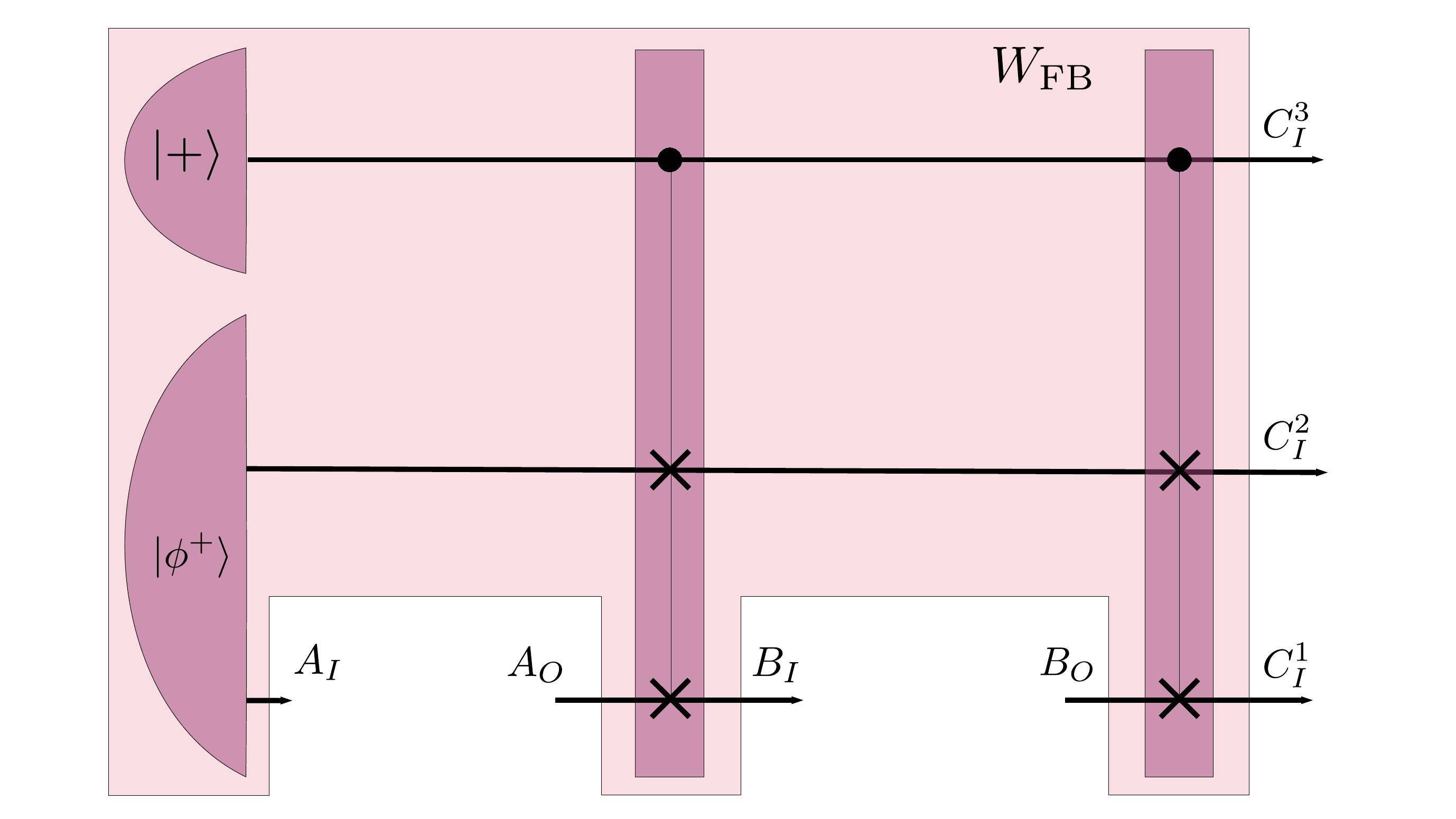}
    \caption{Circuit representation of the process $W_{\text{FB}}$ presented in Ref.\,\cite{Feix2016}. A fixed qubit $\ket{+} = \frac{1}{\sqrt{2}}(\ket{0} + \ket{1})$ in $C_I^{3}$ controls the swap operation from the auxiliary system in $C_I^2$ and $A_O$ to $B_I$ and from $C_I^2$ and $B_O$ to $C_I^1$. As the control qubit is fixed, the swapping operation occurs in the same way in both channels. This combination of operations generate the pure non-classical CCDC process $W_{\text{FB}}$} \label{fig:feix}
\end{figure} 

In Ref.\,\cite{Feix2016}, the authors numerically obtained a lower-bound to the generalized robustness of $W_{\text{FB}}$, using the PPT outer approximation of DC processes, obtaining%
\footnote{Strictly speaking, the authors have evaluated quantity before mentioned here, the non-classicality of causality $\mathcal{C}$, which has a one-to-one relation with the generalized robustness via $R_{G}^{\text{low,PPT}}(W) = \frac{\mathcal{C}(W)}{1+\mathcal{C}(W)}$.}%
\begin{equation}
R_{G}^{\text{low,PPT}}(W_{\text{FB}}) = 0.1855.
\end{equation}
Using our inner approximation method with $N=200$, we could obtain, up to numerical precision, the upper-bound 
\begin{equation}
R_{G}^{\text{up,}\mathcal{E}_{200}}(W_{\text{FB}}) = 0.1855,
\end{equation}
showing that $R_{G}(W_{\text{FB}}) = 0.1855$, up to numerical precision.
For completeness, we have also computed the white noise robustness, obtaining
\begin{equation}
R_{WN}^{\text{low,PPT}} (W_{\text{FB}}) = R_{WN}^{\text{up,}\mathcal{E}_{200}} (W_{\text{FB}}) = 0.3324. 
\end{equation}

\subsection{The role of coherent mixture of causal relations}

	As discussed in this section, previous research has shown that one way to obtain processes with the non-classical CCDC property is by coherently superposing causal relations. For instance, the tripartite processes presented in Ref.\,\cite{Feix2016} and discussed in Section \ref{sec:feix} is constructed in a way to be a coherent superposition of a purely common-cause and a purely direct-cause processes. Also, the bipartite process presented in Ref.\,\cite{MacLean2016} is inspired by a coherent mixture of causal relations which is mathematically formalized by the application of the partial swap operation (see Eq.\,\eqref{eq:partial_swap}).

	Differently from previous works, we have shown in this work that the connection between non-classical CCDC and coherent superpositions of causal relations may be more subtle than it seems at first glance. In particular, \cyan{although the non-classical CCDC process $W_{ddd^2}$ presented in section \ref{sec:ncccdc} may be viewed as a superposition of a process which is initialized in $\ket{00}$ with a process which is initialized in state $\ket{11}$, it also admits a natural interpretation as a process with both common-cause and direct-cause relations simultaneously, without explicitly considering any superposition of causal relations. We then argue that the process $W_{ddd^2}$ does not need to be interpreted as a coherent superposition of causal relations. }

\section{Separable process without a direct-cause explanation}
\label{sec:separableDC}

As mentioned in previous sections, the definition of direct-cause processes (Def.\ref{def:dc}) reminds one of the definition of separable quantum states. More precisely, if we do not impose that the operators $D_i^{A_O/B_I}$ have to respect the quantum channel condition $\tr_{B_I}D_i^{A_O/B_I} = \identity^{A_O}$, equation \eqref{eqdcproc} is precisely the definition of a separable state on the bipartition $A_I|A_OB_I$. We now show that, despite being related, these two definitions are not equivalent.

Consider the process
\begin{align}
	\begin{split}
	\label{eq:Wsep}
	W_\text{SEP} := \frac{1}{2} \Big( & \ketbra{0}{0}^{A_I} \otimes \ketbra{0}{0}^{A_O} \otimes \ketbra{0}{0}^{B_I} \\
	 + & \ketbra{1}{1}^{A_I} \otimes \ketbra{0}{0}^{A_O} \otimes \ketbra{1}{1}^{B_I}  \\
	 + & \ketbra{+}{+}^{A_I} \otimes \ketbra{1}{1}^{A_O} \otimes \ketbra{+}{+}^{B_I}  \\
	 + & \ketbra{-}{-}^{A_I} \otimes \ketbra{1}{1}^{A_O}\otimes \ketbra{-}{-}^{B_I} \Big), 
	 \end{split}
\end{align}
which is a separable operator by construction. First, notice that 
\begin{align}
\tr_{B_I} W_{\text{SEP}} = \frac{1}{2}\identity^{A_I} \otimes \identity^{A_O},
\end{align}
which shows that $W$ is a valid bipartite ordered process. 

The process  $W_\text{SEP}$  can be physically realized by preparing a maximally entangled qubit state between $A_I$ and $\aux$, and using a ``decoder'' described by
\small
\begin{align}
D^{A_O \aux/ B_I} & = \ketbra{0}{0}^{A_O} \otimes \left( \ketbra{00}{00} + \ketbra{11}{11} \right)^{\aux B_I} \nonumber \\
& + \ketbra{1}{1}^{A_O} \otimes \left( \ketbra{++}{++} + \ketbra{--}{--}\right)^{\aux B_I}.
\end{align}
\normalsize
In this way we have 
\begin{align}
W_\text{SEP}=\ketbra{\phi^+}{\phi^+}^{A_I \aux} * D^{A_O \aux / B_I}.
\end{align}
Note that the channel $D^{A_O \aux/B_I}$ can be implemented as follows: first, perform a computational basis measurement on $A_O$. If the outcome is $0$, perform a computational basis measurement on $\aux$ and send the output qubit to $B_I$. If the outcome is $1$, perform a measurement on $\aux$ in the $X$-basis instead. 

\begin{thm}
The bipartite ordered process $W_{\text{SEP}}$ is separable in the bipartition $A_I|A_O B_I$, but is not direct-cause.
\end{thm}

\begin{proof}
Equation \eqref{eq:Wsep} represents $W_{\text{SEP}}$ as a convex combination of product states, ensuring that $W_{\text{SEP}}$ is separable in the bipartition $A_I|A_O B_I$. In order to show that  $W_{\text{SEP}}$ is not a direct-cause process, let us assume that $W_\text{SEP}$ can be written as a convex combination $W_\text{SEP} = \sum_{i} p_i \rho_i^{A_I} \otimes D_i^{A_O/B_I}$, where $\rho_i^{A_I}$ are normalized quantum states and every $D_i^{A_O/B_I}$ satisfy $\tr_{A_I} D_i^{A_O/B_I} = \identity^{A_O}$. 
Note that each $\rho_i^{A_I}$ has non-trivial overlap with, at least, 3 out of the 4 states $\ketbra{0}{0}$, $\ketbra{1}{1}$, $\ketbra{+}{+}$,$\ketbra{-}{-}$. Indeed, let $\rho = \frac{1}{2} \left( \identity + \sum_i \alpha_i \sigma_i \right)$ be an arbitrary state where $\{ \alpha_i \}$ are real numbers that satisfy $\sum_i \alpha_i^2 \leq 1$ and $\sigma_i$ are Pauli matrices. Suppose that $\rho$ has zero overlap with some pure state $\ketbra{\psi}{\psi} =  \frac{1}{2} \left( \identity + \sum_i \beta_i \sigma_i \right)$ with $\sum_i \beta_i^2 = 1$, then we have that $( \vec {\alpha}, \vec{\beta}) = -1$, where $(\cdot, \cdot)$ is the Euclidean inner product. By the Cauchy-Schwarz inequality, we have
\begin{equation}
	 1 = ( \vec {\alpha}, \vec{\beta})^2 \leq ( \vec {\alpha}, \vec{\alpha}) ( \vec {\beta}, \vec{\beta}) =  ( \vec {\alpha}, \vec{\alpha}),
\end{equation}
with equality if and only if $\vec{\alpha}$ is a multiple of $\vec{\beta}$. This shows that $\vec{\alpha} = - \vec{\beta}$ and therefore $\rho$ cannot be orthogonal to any other pure quantum state. 

Let us choose some fixed index $j$ in the sum and suppose, without lack of generality, that $\rho_j$ has non-zero overlap with $\ket{1}, \ket{+}, \ket{-}$. Then, from the above definition of $W_{\text{SEP}}$, we calculate
\begin{align}
 \tr(\ketbra{1}{1}^{A_I} \otimes \ketbra{0}{0}^{A_O} \otimes \ketbra{0}{0}^{B_I} W_{\text{SEP}}) = 0,
\end{align}
From the decomposition $W_\text{SEP} = \sum_{i} p_i \rho_i^{A_I} \otimes D_i^{A_O/B_I}$, we get
\small
\begin{equation}
\sum_i p_i \tr(\ketbra{1}{1} \rho_i^{A_I}) \tr( \ketbra{0}{0}^{A_O} \otimes \ketbra{0}{0}^{B_I} D_i^{A_O/B_I}) = 0.
\end{equation}
\normalsize
By positivity of $D^{A_O/B_I}$ and $\rho^{A_I}$, each term in the sum has to be zero. Since, by assumption, $ \tr(\ketbra{1}{1} \rho_j) \neq 0$, it must be the case that $D_j$ obeys
\begin{equation}
 \tr( \ketbra{0}{0}^{A_O} \otimes \ketbra{0}{0}^{B_I} D_j^{A_O/B_I}) =0.
\end{equation}
Similarly, by calculating other projectors we get
\begin{subequations}
	\begin{align}
	 \tr( \ketbra{1}{1}^{A_O} \otimes \ketbra{-}{-}^{B_I} D_j^{A_O/B_I}) =0, \\
	 \tr( \ketbra{1}{1}^{A_O} \otimes \ketbra{+}{+}^{B_I} D_j^{A_O/B_I}) =0.
	\end{align}
\end{subequations}
From this we get $\tr(\ketbra{1}{1}^{A_O} \otimes \identity^{B_I} D_j^{A_O/B_I}) = 0$, which means that $\tr_{B_I} D_j^{A_O/B_I} \neq \identity^{A_O}$.
\end{proof}

By construction, $W_\text{SEP}$ is separable in the bipartition $A_I|A_O B_I$, implying that $W_\text{SEP}$ has a PPT $k$-symmetric extension for every $k\in\mathbb{N}$ and $R_G^{\text{low, PPT}_k}(W)=R_{WN}^{\text{low, PPT}_k}(W)=0$. This means that the non-classical CCDC property of $W_\text{SEP}$  cannot be certified by any purely entanglement-based criterion, such as the ones explored in Refs.\cite{MacLean2016,Feix2016,Giarmatzi2018}.

	Using the inner and outer approximations presented in Secs.\,\ref{subsec:inndc} and \ref{subsec:outdc} with the family of qubits presented in the Appendix B of Ref.\cite{hirsch17} ($n=171$, which corresponds to $N=2*171^2$ states), we obtain upper and lower bounds for the generalized and white noise robustnesses of $W_{\text{SEP}}$
\begin{subequations}	
\begin{align}\label{eq:believe}
R_{G}^{\text{up}, \mathcal{E}_N}(W_{\text{SEP}}) = R_{G}^{\text{low}, \hat{\mathcal{E}}_N}(W_{\text{SEP}})=0.1465, \\
R_{WN}^{\text{up}, \mathcal{E}_N}(W_{\text{SEP}}) = R_{WN}^{\text{low}, \hat{\mathcal{E}}_N}(W_{\text{SEP}})=0.2930,
\end{align}
\end{subequations}
with equality holding up to numerical precision.

In Ref.\,\cite{Giarmatzi2018}, it was conjectured that the set of separable processes and the set of processes without quantum memory are not the same, the latter being a strict subset of the first. Since, the definitions of bipartite processes without quantum memory and bipartite direct-cause processes are equivalent (see Appendix \ref{sec:classmemdc}),  we have then proven the conjecture presented in Ref.\,\cite{Giarmatzi2018} by explicitly constructing the bipartite separable process $W_{\text{SEP}}$, which cannot be realized by processes without quantum memory.

It is interesting to point that our numerical methods allowed us to obtain a relatively high robustness of $R_{WN}(W_{\text{SEP}}) = 0.2930$, but techniques exclusively based on entanglement would lead to the trivial lower bound $R^\text{sep}_{WN}(W_{\text{SEP}})\geq0$. Since $R_{WN}(W_{\text{SEP}}) = 0.2930$  is considerably greater than zero, we see that the approximating the set of direct-cause processes by separable processes may lead to very unsatisfactory results.

\section{Certifying non-classical CCDC on PPT processes}
\label{sec:WPPT}

	In this section, we present an example of a non-classical CCDC process which has $R_G^{\text{low,PPT}}(W)=R_{WN}^{\text{low,PPT}}(W)=0$, \emph{i.e.}, its non-classical CCDC property cannot be certified by the PPT approximation used in Refs.\,\cite{MacLean2016, Feix2016}. Such a process can be obtained by exploiting a class of entangled states with positive partial transpose, presented in Ref.\,\cite{Horodecki1997}. The class of states of our interest is
\small
\begin{align}
\begin{split}
\label{eq:horodeckist}
    & \rho_a^{2\times4} := \\
    & \frac{1}{7a+1}
    \begin{bmatrix}
    a & 0 & 0 & 0 & 0 & a & 0 & 0 \\
    0 & a & 0 & 0 & 0 & 0 & a & 0 \\ 
    0 & 0 & a & 0 & 0 & 0 & 0 & a \\ 
    0 & 0 & 0 & a & 0 & 0 & 0 & 0 \\ 
    0 & 0 & 0 & 0 & \tfrac{1}{2}(1+a) & 0 & 0 & \tfrac{1}{2}\sqrt{1-a^2} \\
    a & 0 & 0 & 0 & 0 & a & 0 & 0 \\
    0 & a & 0 & 0 & 0 & 0 & a & 0 \\
    0 & 0 & a & 0 & \tfrac{1}{2}\sqrt{1-a^2} & 0 & 0 & \tfrac{1}{2}(1+a)\end{bmatrix},
    \end{split}
\end{align}
\normalsize
being entangled for $a \in (0,1)$ and separable for $a=0$ or $a=1$.

We now set $a=\frac{1}{2}$ and use $\rho_{\frac{1}{2}}^{2 \times 4}$ to define:
\begin{equation} \label{eq:Wppt}
    W_{\textrm{PPT}} := d_{A_O} \cdot \rho_{\frac{1}{2}}^{2 \times 4},
\end{equation}
with $d_{A_O}=2$. $W_{\textrm{PPT}}$ is a valid bipartite ordered process, as it satisfies every condition from Eqs.\,\eqref{eq:ordproc} by direct inspection. Also, it has dimensions $d_{A_I}{=} d_{A_O} {=} d_{B_I} = 2$ and has bound entanglement in the bipartition $A_I|A_O B_I$.

`We will now ensure that $W_{\textrm{PPT}}$ is a non-classical CCDC process by using a better approximation $\mathcal{L}_{\text{DC}}^{\text{out,PPT}_k}$. In particular, we set $k=2$, then we obtain
\begin{align}
\begin{split}
\label{eq:wpptoutrobs}
R_G^{\text{low,PPT}_2}(W_{\textrm{PPT}}) & = 0.0083, \\
R_{WN}^{\text{low,PPT}_2}(W_{\textrm{PPT}}) & = 0.0230.
\end{split}
\end{align}
For $k = 3$, the robustnesses do not change, which indicates that using greater values of $k$ does not improve the values of generalized and white noise robustnesses.

When using the inner and outer approximations presented in Secs.\,\ref{subsec:inndc} and \ref{subsec:outdc} with the family of qubits presented in the Appendix B of Ref.\cite{hirsch17} ($n=171$, which corresponds to $N=2*171^2$ states), up to numerical precision, we obtain
\begin{equation}
\begin{split}
\label{eq:wpptinrobs}
R_G^{\text{up},\mathcal{E}_{10^4}}(W_{\textrm{PPT}}) &= R_G^{\text{low},\hat{\mathcal{E}}_{10^4}}(W_{\textrm{PPT}}) = 0.1085, \\
R_{WN}^{\text{up},\mathcal{E}_{10^4}}(W_{\textrm{PPT}}) & = R_{WN}^{\text{low},\hat{\mathcal{E}}_{10^4}}(W_{\textrm{PPT}}) = 0.2782.
\end{split}
\end{equation}
	We verify that the lower bounds for the robustnesses obtained with the entanglement criterium in Eqs.\,\eqref{eq:wpptoutrobs} are rather loose in comparison to the actual robustnesses values from Eqs.\,\eqref{eq:wpptinrobs}. This example also illustrates the limitations of certifying non-classical CCDC solely based on entanglement criteria.

\section{Summary of generalized and white noise robustnesses}
\label{sec:summary}

\begin{table*}[t]
\centering
{\renewcommand{\arraystretch}{2}
\begin{tabular}{|c|c|c|c|c|}
\hline
Process (Eq.)                       & $R_G$         & $R_G^{\text{low, PPT}_{k=2}}$ & $R_{WN}$ & $R_{WN}^{\text{low, PPT}_{k=2}}$ \\ \hline
$W_{333^2}$\,\eqref{eq:Wddd2}       & $\frac{2}{3}$ & \green{$\frac{2}{3}$} &$0.9529$      & \green{$0.9529$}\\ \hline
$W_{222^2}$\,\eqref{eq:Wddd2}       & $\frac{1}{2}$ & \green{$\frac{1}{2}$} &$0.8421$      & \green{$0.8421$}\\ \hline
$W_{222}$\,\eqref{eq:W222}          & $\frac{1}{2}$ & \green{$\frac{1}{2}$} &$\frac{2}{3}$ & \green{$\frac{2}{3}$}\\ \hline
$W_{\text{MRSR}}$\,\eqref{eq:W_MRSR}& $0.3506$      & \green{$0.3506$}      &$0.5000$      & \green{$0.5000$}\\ \hline
$W_{\text{FB}}$\,\eqref{eq:W_FB}    & $0.1855$      & \green{$0.1855$}     &$0.3324$      & \green{$0.3324$}\\ \hline
$W_{\text{SEP}}$\,\eqref{eq:Wsep}   & $0.1465$      & \red{$0$}           &$0.2930$      & \red{$0$} \\ \hline
$W_{\text{PPT}}$\,\eqref{eq:Wppt}   &$0.1085 $      & \red{$0.0083$}      & $ 0.2782$    & \red{$0.0230$}  \\ \hline
\end{tabular}
}
\caption{
Table presenting generalized and white noise robustnesses for every process analyzed in this work. Values represented in fractions were obtained by mathematical theorems and coincide with SDP optimization. Values with decimal digits were obtained only via SDP optimization, where our upper and lower bounds are identical up to $4$ decimals. The lower bounds obtained by the approximations of the CCDC set based on entanglement are the values provided in columns $R_G^{\text{low, PPT}_k}$ and $R_{WN}^{\text{low, PPT}_k}$. Lower-bounds that match the actual robustnesses are highlighted in green, while lower bounds that are rather different from the actual robustnesses are highlighted in red. We can observe that entanglement based criteria could never detect $W_\text{SEP}$ as a non-classical CCDC process, while the PPT $k$-symmetric extension bound for $W_{\text{PPT}}$ provides the loose lower bounds $R_G^{\text{low, PPT}_k}(W_\text{PPT})\geq 0.0083$ $R_{WN}^{\text{low, PPT}_k}(W_\text{PPT})\geq 0.0230$, which are values obtained both with $k=2$ and $k=3$.
}
\label{table:robsummary}
\end{table*}

In previous sections, we presented several examples of non-classical CCDC processes with different values of generalized and white noise robustness.
Table \ref{table:robsummary} summarizes the non-classical CCDC robustnesses of several processes presented in this work and compare the actual value of robustness with the values obtained with methods based on entanglement criteria.

\section{Conclusions}

	In this work, we have introduced a class of bipartite ordered processes, and a process with dimension of three qubits, which are maximally robust against general noise and very likely to be the most robust against white noise for the qubit case. This class of processes can be implemented by preparing a pair of maximally entangled states and an identity channel, admitting a natural interpretation of a process with both common-cause and direct-cause relations simultaneously. Hence, in contrast to previously known non-classical CCDC processes \cite{Ried2015,Feix2016}, the class presented here does not require either the construction or the interpretation directly based on coherent superposition of causal relations.

	Several analytical results proved in this work employed general convex analysis arguments based on witness hyper-planes, combined with entanglement theory concepts, such as entanglement breaking channels. We believe that the techniques developed here may find applications in related problems. 
	
	 We have also presented a systematic semi-definite approach to characterize the set of non-classical CCDC processes. In particular, we provided a hierarchy of inner and outer approximations that converge to the set of classical CCDC processes, and an entanglement-based hierarchy which, despite not converging to the set of CCDC processes, provides us several tight and non-trivial bounds.

	In order to tackle situations where we could not prove the value of the highest robustness of a given scenario analytically, we constructed a heuristic see-saw method, which provided numerical evidence of the highest robustnesses attainable on such scenario, also providing a valid lower-bound for the value of the highest robustness.	

	Finally, we have shown that, although all bipartite processes that are entangled in the bipartition $A_I|A_O B_I$ do not have a direct-cause decomposition, the converse does not hold. Our proof consists in explicitly constructing a process which is separable on the bipartition $A_I|A_O B_I$, but does not have a direct-cause decomposition. Since bipartite processes without quantum memory are equivalent to bipartite direct-cause ones, our results prove a conjecture first raised in Ref.\,\cite{Giarmatzi2018} and contributes towards a better understanding of the particularities of quantum memory, spacial entanglement and temporal entanglement \cite{Milz2017,Giarmatzi2018,Taranto2019,Taranto2020,Milz2020A}.

		All SDP optimization problems presented in this manuscript were implemented using MATLAB, the convex optimization package Yalmip \cite{yalmip} and CVX \cite{cvx}, the solvers MOSEK, SeDuMi and SDPT3 \cite{MOSEK, sedumi, SDPT3}, and the toolbox for quantum information QETLAB \cite{QETLAB}. All our codes are available in the public repository \cite{ccdc} and can be freely used under the GNU Lesser General Public License v3.0.

\begin{acknowledgments}

We are grateful to Mateus Araújo, Jessica Bavaresco, Rafael Chaves, Simon Milz and Philip Taranto for interesting discussions. MN, RV and TM acknowledge financial support by the Brazilian agencies INCT-IQ (National Institute of Science and Technology for Quantum Information), FAPEMIG, CNPq, CAPES and Instituto Serrapilheira. MTQ has received funding from the European Union’s Horizon 2020 research and innovation programme under the Marie Skłodowska-Curie grant agreement No 801110 and the Austrian Federal Ministry of Education, Science and Research (BMBWF). It reflects only the authors' view, the EU Agency is not responsible for any use that may be made of the information it contains. MTQ and PG acknowledge the support of the Austrian Science Fund (FWF) through the SFB project "BeyondC", a grant from the Foundational Questions Institute (FQXi) Fund and a grant from the John Templeton Foundation (Project No. 61466) as part of the Quantum Information Structure of Spacetime (QISS) Project (qiss.fr). The opinions expressed in this publication are those of the authors and do not necessarily reflect the views of the John Templeton Foundation. Research at Perimeter Institute is supported in part by the Government of Canada through the Department of Innovation, Science and Economic Development Canada and by the Province of Ontario through the Ministry of Colleges and Universities.

\end{acknowledgments}

\nocite{apsrev42Control} 
\bibliographystyle{0_MTQ_apsrev4-2_corrected}

\bibliography{library.bib}

\onecolumn\newpage

\appendix

\section{Ordered processes without quantum memory are equivalent to direct-cause processes on the bipartite case}
\label{sec:classmemdc}

	We now present the definition of ordered (non-Markovian) processes without quantum memory, first introduced in Ref.\cite{Giarmatzi2018}.
	
\begin{defn}[Process without quantum memory\cite{Giarmatzi2018}]
	A linear operator $W\in \mathcal{L}(A_I\otimes A_O \otimes B_I)$ is a bipartite ordered process without quantum memory if it can be written as
\begin{equation}
	W = \rho_\text{SEP}^{A_I\aux}*D^{A_O \aux/B_I},
\end{equation}
where $\rho_\text{SEP}^{A_I \aux}$ is a separable state and  $D^{A_O\aux / B_I}$ is the Choi operator of a quantum channel.
\end{defn}

We now show that all bipartite processes without quantum memory are bipartite direct-cause, \textit{vice-versa}. This proof was first presented in Appendix A.3.2 of Ref.\,\cite{Giarmatzi2018}, but we reproduce it here for the sake of completeness.

\begin{thm}[Appendix A.3.2 of Ref.\cite{Giarmatzi2018}] \label{theo:classical}
	A bipartite ordered process $W$ is a process without quantum memory if and only if $W$ is direct-cause.
\end{thm}

\begin{proof}
	We first show that every process without quantum memory is direct-cause.
	Since $\rho_{\text{SEP}}^{A_I\aux}$ is separable, there exists $\rho_i^{A_I}$ and $\sigma_i^{\aux}$ and some probabilities $p_i$  such that $ \rho_{\text{SEP}}^{A_I\aux} = \sum_i p_i\rho_i^{A_I} \otimes \sigma_i^{\aux} $. Thus, we can write

	\begin{align}
\left(\sum_i p_i \rho_i^{A_I} \otimes \sigma_i^{\aux}\right) * D^{\aux A_O/B_I} 
	= \sum_i p_i \rho_i^{A_I} \otimes D_i^{A_O/B_I},
	\end{align}
where $D_i^{A_O/B_I}:= \sigma_i^{\aux} *  D^{\aux A_O / B_I}$ are valid quantum channels, since $D_i^{A_O/B_I} \succeq 0$ and $\tr_{B_I}(D_i^{A_O/B_I})= \identity^{A_O}$.

	Now, we need to show that every direct-cause process is a process without quantum memory in the bipartite case. Let us assume that $W$ is direct-cause. Then $W$ can be written as
\begin{equation}
	 W =\sum_i p_i \rho_i^{A_I} \otimes D_i^{A_O/B_I}
\end{equation}
for some states $\rho_i^{A_I}$ and channels $D_i^{A_O/B_I}$. Now, let us define the separable state $\sigma^{A_I \aux}:= \sum_i p_i \rho_i^{A_I} \otimes \ketbra{i}{i}^{\aux}$, and the quantum channel $D^{\aux A_O/B_I}:= \sum_i \ketbra{i}{i}^{\aux} \otimes D_i^{A_O/B_I}$. Direct calculation shows that $\sigma^{A_I,\aux}*D^{\aux A_O/B_I}=W$, ensuring that $W$ is a process without quantum memory.
\end{proof}

\section{A tight upper bound for the generalized robustness}
\label{app:robustness_bound}

In this section, we prove that the generalized robustness of any process is upper-bounded by $R_G(W) \leq  1 - \frac{1}{d_{A_I}}$. This bound is saturated by the processes $W_{ddd^2}$ (Eq.\,\eqref{eq:Wddd2}) for every dimension $d$ and by $W_{222}$ (Eq.\,\eqref{eq:W222}). 

\begin{lem} \label{lemma:ent_break}
	Let $W$ be a bipartite ordered process. If $\map{\Lambda}:\mathcal{L}(A_I)\to\mathcal{L}(A_I)$ is an entanglement breaking channel, the process
	$\map{\Lambda}^{A_I}\otimes \map{\identity}^{A_OB_I}(W)$ is direct-cause.
\end{lem}
\begin{proof}
	By definition, any bipartite ordered process can be written as $W=\rho^{A_I \aux} * D^{\aux A_O / B_I}$. Since $\map{\Lambda}$ is entanglement breaking, it holds that $\map{\Lambda}^{A_I}\otimes \map{\identity}^{\aux}(\rho^{A_I \aux})$ is a separable state. Therefore, we can use the same argument presented in the proof of theorem \ref{theo:classical} to ensure that $\map{\Lambda}^{A_I}\otimes \map{\identity}^{A_OB_I}(W)$ is a direct-cause process.
\end{proof}

\begin{lem} \label{lemma:lambda}
Let $\map{\Lambda}:\mathbb{C}_d\to\mathbb{C}_d$ be a quantum channel,
$\omega=e^\frac{2\pi\sqrt{-1}}{d}$, and
$Z:=\sum_{i=0}^{d-1}\omega^i \ketbra{i}{i} $ be the $d$-dimensional clock operator. The channel
\begin{equation}
\map{\Lambda}(\rho) = \frac{1}{d} \sum_{k=0}^{d-1} Z^k \rho Z^{-k}
\end{equation}
is an entanglement breaking channel.
\end{lem}

\begin{proof}
A necessary and sufficient condition~\cite{Horodecki2003} for $\map{\Lambda}$ to be entanglement-breaking is that its Choi operator $\Lambda$ is separable between input and output spaces. We now show that $\Lambda$ is separable by
\begin{subequations}
	\begin{align}
	\label{eq:lambda_EB}
	\Lambda :=& \sum_{ab} \ketbra{a}{b} \otimes \map{\Lambda}(\ketbra{a}{b}) \\
	         =& \frac{1}{d}\sum_{abk} \ketbra{a}{b} \otimes Z^k(\ketbra{a}{b})Z^{-k} \\
	         =& \frac{1}{d} \sum_{abkij} \ketbra{a}{b} \otimes \ketbra{i}{i} \omega^{ik}(\ketbra{a}{b})\ketbra{j}{j} \omega^{-jk} \\
	         =& \frac{1}{d} \sum_{kij} \ketbra{i}{j} \otimes \ketbra{i}{j} \omega^{ik} \omega^{-jk} \\
	 =& \frac{1}{d} \sum_{ij} \ketbra{i}{j} \otimes \ketbra{i}{j} \sum_k \omega^{ik} \omega^{-jk} \\
	 =& \frac{1}{d} \sum_{ij} \ketbra{i}{j} \otimes \ketbra{i}{j} \sum_k \omega^{k(i-j)} \\
	  =& \frac{1}{d} \sum_{ij} \ketbra{i}{j} \otimes \ketbra{i}{j} d\delta_{ij} \\
	  =& \sum_{i} \ketbra{i}{i} \otimes \ketbra{i}{i}. 
	\end{align}
\end{subequations}
\end{proof}

\begin{thm}\label{theo:upper}
The generalized robustness of a bipartite ordered process $W\in \mathcal{L}(A_I\otimes A_O\otimes B_I)$ is upper-bounded by $R_G(W) \leq 1 - \frac{1}{d_{A_I}}$.
\end{thm}

\begin{proof}
Let $\map{\Lambda}: \mathcal{L}(A) \to \mathcal{L}(A)$ be the entanglement-breaking channel defined in Lemma \ref{lemma:lambda}. Note that since $Z^0=\identity$, the action of $\map{\Lambda}$ can be written as
\begin{subequations}	
	\begin{align} \label{eq:lambda}
	\map{\Lambda}(\rho) &=  \frac{1}{d} \sum_{k=0}^{d-1} Z^k \rho Z^{-k} \\
	 &= \frac{1}{d} \rho + \frac{1}{d} \sum_{k=1}^{d-1} Z^k \rho Z^{-k} \\
		 &= \frac{1}{d} \rho + \left(1-\frac{1}{d}\right) \map{\Lambda}_{\setminus}(\rho),
	\end{align}
\end{subequations}
where $\map{\Lambda}_\setminus(\rho):=\frac{1}{d-1}\sum_{k=1}^{d-1} Z^k \rho Z^{-k}$ is a valid quantum channel.

	Lemma \ref{lemma:ent_break} states that, for any bipartite ordered process $W$, the process $\map{\Lambda}^{A_I}\otimes \map{\identity}^{A_OB_I}(W)$ is direct-cause, thus being CCDC. Hence, by making use of Eq.\,\eqref{eq:lambda}, we see that the resulting process
\begin{equation}
 \frac{1}{d_{A_I}} W + \left(1-\frac{1}{d_{A_I}}\right) \map{\Lambda}_{\setminus}^{A_I}\otimes \map{\identity}^{A_OB_I}(W)
\end{equation}
is CCDC, with $\map{\Lambda}_{\setminus}^{A_I}\otimes \map{\identity}^{A_OB_I}(W)$ being a valid bipartite ordered process. 

By analyzing the definition of generalized robustness presented in Eq.\,\eqref{eq:genrob}, we see that setting $\Omega=\map{\Lambda}_{\setminus}^{A_I}\otimes \map{\identity}^{A_OB_I}(W)$ ensures that the relation
\begin{equation}
R_G(W) \leq 1 - \frac{1}{d_{A_I}}
\end{equation}
holds for any bipartite ordered process $W$.
\end{proof}

\section{An upper bound for the white noise robustness}
\label{sec:upper_WNR}

In this section, we present an upper bound for the white noise robustness. Differently from the generalized robustness case, this bound is not tight, but is useful for proving the strong duality relation for the problem of evaluating the white noise robustness for non-classical CCDC processes in Section \ref{sec:strong_duality}.

\begin{thm} \label{lemma:epsilon_ball}
The white noise robustness of a bipartite ordered process $W\in \mathcal{L}(A_I\otimes A_O\otimes B_I)$ is upper-bounded by $R_{WN}(W) \leq 1-\frac{1}{d_{A_I}d_{A_O} d_{B_I} + 1}$.
\end{thm}
\begin{proof}
	The depolarizing channel $\map{D}_\eta(\rho):= (1-\eta) \rho + \eta \frac{\identity}{d}$ is known to be entanglement breaking when $\eta\geq \frac{d}{d+1}$ \cite{Horodecki1997a, Horodecki2003}. Hence, Lemma \ref{lemma:ent_break} ensures that, for any bipartite ordered process
 $W\in \mathcal{L}(A_I\otimes A_O\otimes B_I)$, the process
\begin{equation}
	\map{\Lambda}^{A_I}\otimes \map{\identity}^{A_OB_I}(W) = 
	\frac{1}{d_{A_I}+1} W  + \frac{d_{A_I}}{d_{A_I}+1} \left( \frac{\identity^{A_I}}{d_{A_I}} \otimes \tr_{A_I}(W) \right)
\end{equation}	
	 is direct-cause. We now define the operator 
\begin{equation}
W':= 
\frac{\identity^{A_I}}{d_{A_I}} \otimes
 \frac{\left(d_{A_O} \identity^{A_OB_I} - \tr_{A_I}(W)^{A_OB_I} \right)} {d_{A_O}d_{B_I}-1}, 
 \end{equation}
which is a valid direct-cause process by direct inspection. Taking a convex combination of $\map{\Lambda}^{A_I}\otimes \map{\identity}^{A_OB_I}(W)$ and $W'$, we obtain
\small
	\begin{align}
	\begin{split}
			& q \map{\Lambda}^{A_I}\otimes \map{\identity}^{A_OB_I}(W) + (1-q)W' \\
			 = & q\left(\frac{1}{d_{A_I}+1} W + \frac{d_{A_I}}{d_{A_I}+1} \left( \frac{\identity^{A_I}}{d_{A_I}} \otimes \tr_{A_I}(W) \right)\right)	 + (1-q) \frac{\identity^{A_I}}{d_{A_I}} \otimes
 \frac{\left(d_{A_O} \identity^{A_OB_I} - \tr_{A_I}(W)^{A_OB_I} \right)} {d_{A_O}d_{B_I}-1} \\
   = & \frac{q}{d_{A_I}+1} W  +
 \left( \frac{(1-q)d_{A_O}d_{B_I}}{d_{A_O}d_{B_I}-1} \right) \frac{\identity^{A_I}}{d_{A_I}} \otimes
 \frac{ \identity^{A_OB_I}}{d_{B_I}}  + \left( 
\frac{q d_{A_I}}{d_{A_I}+1} -  \frac{(1-q)}{(d_{A_O}d_{B_I}-1)} 
\right) \frac{\identity^{A_I}}{d_{A_I}} \otimes \tr_{A_I}(W), 
\end{split}
	\end{align}
\normalsize
	which is direct-cause by construction. Now, by setting 
		$q = \frac{d_{A_I} + 1}{d_{A_I}d_{A_O} d_{B_I} + 1}$,
 we obtain
\begin{align}
\frac{q d_{A_I}}{d_{A_I}+1} -  \frac{(1-q)}{(d_{A_O}d_{B_I}-1)}  = 0,
\end{align}
and the process 

\begin{equation}
\frac{1}{d_{A_I}d_{A_O} d_{B_I} + 1}W + 
\left(1-\frac{1}{d_{A_I}d_{A_O} d_{B_I} + 1} \right)\frac{\identity}{d_{A_I}d_{B_I}}
\end{equation}

is guaranteed to be direct-cause.
\end{proof}

\section{Strong duality for CCDC robustness problems}
\label{sec:strong_duality}

\begin{thm}
	The convex optimization problems for non-classical CCDC generalized robustness (Eq.\,\eqref{eq:genrob}) and non-classical CCDC white noise robustness (Eq.\,\eqref{eq:wnrob}) satisfy strong duality.
\end{thm}

\begin{proof}
We recall that every convex optimization problem admitting a strictly feasible solution \textit{i.e.,} 
all equality constraints are satisfied, and all inequality constraints are strictly satisfied,
necessarily satisfies strong duality (Slater condition \cite{Boyd2004}). 

From Theorem \ref{lemma:epsilon_ball} we see that for any value $1-\frac{1}{d_{A_I}d_{A_O} d_{B_I} + 1}<r<1$, for any process $W$, the process $ \Omega := (1-r) W + r \frac{\identity}{d_{A_I} d_{B_I}}$ is a strictly feasible solution, ensuring that both robustness problems respect strong duality.
\end{proof}

\section{SDPs for CCDC separability}
\label{sec:sdp}

In this section, we present the explicit forms of the SDPs mentioned in Section \ref{sec:cetifnccdc} to obtain the generalized and white noise robustnesses of a bipartite ordered process $W$. 

Consider the problem of obtaining the PPT$_k$ generalized robustness of a process $W$, \emph{i.e.}, finding the robustnesses of a process against its worst noise $\Omega$, so the resulting combination process lies in $\mathcal{L}_{\AB}^{\text{out,PPT}_k}$. This is represented by the following optimization program:
\begin{subequations}
\label{eqs:genrobst}
	\begin{align}
	& R_G^{\text{low,PPT}_k}(W) := \textrm{min}  \hspace{.2cm}   r  \\
	             & \textrm{s.t.}   \hspace{.2cm}   (1-r)W + r \Omega = q \Wcc + (1-q) \Wdc^{\text{PPT}_k},\\
	                             &  \hspace{.8cm} 0 \leq r \leq 1, \\        
	                              & \hspace{.8cm} 0 \leq q \leq 1, \\
	                              & \hspace{.8cm} \Wcc \in \mathcal{L}_{\text{CC}}, \ \Wdc^{\text{PPT}_k} \in \mathcal{L}_{\text{DC}}^{\text{out,PPT}_k}, \ \Omega \in \mathcal{L}_{\AB},
	\end{align}
\end{subequations}
with the process $W_{\text{DC}}^{\text{PPT}_k}$ having a $k$-symmetric PPT extension $W_{\text{DC}}^{A_I^{\otimes k}| A_O B_I} \in \mathcal{L}(A_I^{\otimes k} \otimes A_O \otimes B_I)$ \cite{Doherty2001,Doherty2003}. This means that  there exists a positive semi-definite operator $W_{\text{DC}}^{A_I^{\otimes k}| A_O B_I}$ such that  

\begin{align}
 \tr_{A_I^{\otimes k-1}}(W_{\text{DC}}^{A_I^{\otimes k}| A_O B_I})  & = \Wdc^{\text{PPT}_k} \\
 W_{\text{DC}}^{A_I^{\otimes k}| A_O B_I} & = (P_{\text{sym}}^{A_I^{\otimes k}} \otimes \identity^{A_O B_I}) W_{\text{DC}}^{A_I^{\otimes k}| A_O B_I} (P_{\text{sym}}^{A_I^{\otimes k}}\otimes \identity^{A_O B_I}) \label{subeq:k-sym}
\end{align}
with positive partial transposition on $A_I^{\otimes k}$, that is
\begin{equation}
\left( W_{\text{DC}}^{{A_I^{\otimes k}| A_O B_I}}\right)^{T_{A_{I}}} \succeq 0,
\end{equation}
and where $P_{\text{sym}}^{A_I^{\otimes k}}$ is the projector onto the symmetric subspace of $A_I^{\otimes k}$.
We point that, in practice, instead of Eq.\,\eqref{subeq:k-sym}, it is more advantageous to impose the Bose k-symmetric extension condition, that is, to impose
\begin{equation}
W_{\text{DC}}^{A_I^{\otimes k}| A_O B_I}  = (P_{\text{sym}}^{A_I^{\otimes k}} \otimes \identity^{A_O B_I}) W_{\text{DC}}^{A_I^{\otimes k}| A_O B_I}.
\end{equation}
This follows from the fact that, computationally, imposing the Bose k-symmetric extension condition instead of the k-symmetric extension condition does not add any complexity to the problem, but the Bose k-symmetric extension condition detects more entangled states than the standard k-symmetric one \cite{Navascues2009}.

Since Eqs.\,\eqref{eqs:genrobst} contains products of optimization variables, such as $q\Wcc$, the optimization program above is not linear, hence not an SDP. This issue can be circumvented by rewriting the problem in the following equivalent form:
\begin{subequations}
	\begin{align}
	& R_G^{\text{low, PPT}_k}(W) := \textrm{min}  \hspace{.2cm}   \tr \left( \frac{\overline{\Omega}}{d_{A_O}}\right)  \\
	             & \textrm{s.t.}   \hspace{.2cm}   \left(1-\tr \left(\frac{\overline{\Omega}}{d_{A_O}}\right) \right)W + \overline{\Omega} = \overline{\Wcc} +  \overline{\Wdc},\\
	             				 & \hspace{.8cm} \overline{\Wcc} = \rho^{A_I B_I} \otimes \identity^{A_O} \\
	                             & \hspace{.7cm} \left(\overline{\Wdc}^{{A_I^{\otimes k}| A_O B_I}}\right)^{T_{A_{I}}} \succeq 0, \label{subeq:pptcond} \\
	                             & \hspace{.8cm} \overline{\Wdc} = L_{\AB}(\overline{\Wdc}), \ \overline{\Omega} = L_{\AB}(\overline{\Omega}),  \\                             
	                             & \hspace{.8cm} \overline{\Omega}, \ \overline{\Wdc}, \ \rho^{A_I B_I}\succeq 0.   
	\end{align}
\end{subequations}

In the above SDP we have used the fact that valid processes should respect $\tr(W)=d_{A_O}$ to embed the scalar variables into the operators. This means that the new variables relate to the variables from Eqs.\,\eqref{eqs:genrobst} with
\begin{equation}
\overline{\Omega} = r\Omega, \ \overline{\Wcc} = q \Wcc, \ \overline{\Wdc} = (1-q)\Wdc^{\text{PPT}_k}.
\end{equation}

Analogously, the same steps are applied to the PPT white noise robustness, thus being
\begin{subequations}
	\begin{align}
	& R_{WN}^{\text{low,PPT}_k}(W) := \textrm{min}  \hspace{.2cm}   r  \\
	             & \textrm{s.t.}   \hspace{.2cm}   (1-r)W + r \frac{\identity^{A_I A_O B_I}}{d_{A_I} d_{B_I}}= \overline{\Wcc} +  \overline{\Wdc},\\    
	                             & \hspace{.8cm} \overline{\Wcc} = \rho^{A_I B_I} \otimes \identity^{A_O}\\
	                             & \hspace{.7cm} \left(\overline{\Wdc}^{{A_I^{\otimes k}| A_O B_I}}\right)^{T_{A_{I}}} \succeq 0,  \\                          
	                             & \hspace{.8cm} \overline{\Wdc}, \ \rho^{A_I B_I}\succeq 0.   
	\end{align}
\end{subequations}

From the Lagrangian of the SDPs, we can obtain their dual problems, which generate the optimal non-classical CCDC witnesses $S$ for the given process $W$ \cite{Boyd2004}. For instance, considering $k=1$, the dual form%
\footnote{Strictly speaking, the dual objective function is to ``maximize $-\tr(SW)$'', which results in $R(W) = \tr(SW)$. Here, we adopted a convention of changing the sign of $S$ and replacing this objective function by ``minimize $\tr(SW)$''. With this convention, the operator $S$ is a non-classical CCDC witnesses as defined in Section \ref{subsec:witness}, that is, it  satisfies $\tr(SW_{\text{CCDC}}) \geq 0$ for all classical CCDC processes $W_{\text{CCDC}}$. 
}
of the PPT generalized robustness is given by
\begin{subequations}
\label{eqs:dualgenrob}
	\begin{align}
	&  \textrm{min}  \hspace{.2cm}   \tr \left( S W \right)  \\
	             & \textrm{s.t.}   \hspace{.5cm}  \identity^{A_I A_O B_I}\left( 1 + \tr(S W)\right) - d_{A_O} S \succeq 0 \label{subeq:normwitgen}\\    
	                             & \hspace{1cm} S - S_{\text{DC}} + S_{\text{DC}}^{\perp}  \succeq 0 \label{subeq:witgenproj}\\
	                             & \hspace{1cm} L_{\AB}(S_{\text{DC}}^{\perp})= 0, \\
	                             & \hspace{1cm} S_{\text{DC}}^{T_{A_I}} \succeq 0, \label{subeq:witgendc}\\
	                             & \hspace{1cm} \tr_{A_O}(S) \succeq 0, \label{subeq:witgencc}
	\end{align}
\end{subequations}
which is the one we use in our heuristic see-saw presented in Section \ref{sec:seesaw}. The variable $S_\text{DC}^{\perp}$ in Eq.\,\eqref{subeq:witgenproj} is associated with an orthogonal projection onto $\mathcal{L}_{\AB}$. Also, the same methods presented in Section \ref{sec:strong_duality} to prove that the robustness optimization problem satisfies strong duality can be used to show that this upper bound problem also respects strong duality. We, then, have $R_{G}^{\text{low, PPT}} = - \tr (S W)$, when $S$ is the optimal witness for $W$.

We can see that Eq.\,\eqref{subeq:witgenproj} together with Eqs.\,\eqref{subeq:witgendc} and \eqref{subeq:witgencc} ensure that $S$ is a non-classical CCDC witness. Also, Eq.\,\eqref{subeq:normwitgen} corresponds to the normalization condition, which determines that it is a witness for the generalized robustness measure.

Analogously, the dual form of the PPT white noise robustness is given by
\begin{subequations}
\label{eqs:dualwnrob}
	\begin{align}
	&  \textrm{min}  \hspace{.2cm}   \tr \left( S W \right)  \\
	             & \textrm{s.t.}   \hspace{.5cm}  \frac{\tr(S)}{d_{A_I} d_{B_I}} - \tr(SW) \leq 1, \label{subeq:normwitwn}\\    
	                             & \hspace{1cm} S - S_{\text{DC}}^{T_{A_I}} + S_{\text{DC}}^{\perp}  \succeq 0 \label{subeq:ccdcwit}\\
	                             & \hspace{1cm}  L_{\AB}(S_{\text{DC}}^{\perp}) = 0, \\
	                             & \hspace{1cm} S_{\text{DC}} \succeq 0, \\
	                             & \hspace{1cm} \tr_{A_O}(S) \succeq 0,
	\end{align}
\end{subequations}
where Eq.\,\eqref{subeq:normwitwn} is the normalization condition for the witness $S$, which corresponds to the white noise robustness measure.

A similar method can be used to characterize our inner approximation  $\mathcal{L}_{\text{CCDC}}^{\text{in},\mathcal{E}_N}$. In this case we just need to absorb the probabilities $p_i$ of Eq.\,\eqref{subeq:pptcond} into the quantum channel $D_i$ to obtain $\overline{D_i}:=p_i D_i$. This allows us to replace the constraint of Eq.\,\eqref{subeq:pptcond} by
\begin{subequations}
\begin{align}
& \Wdc   = \sum_{i=1}^{N}  \ketbra{\psi_i}{\psi_i}^{A_I} \otimes \overline{D_i}^{A_O/B_I}, \\
& \overline{D_i}^{A_O/B_I}  \succeq 0, \\
& \tracerep{B_I}{\overline{D_i}^{A_O/B_I}}  = \tracerep{A_O B_I}{\overline{D_i}^{A_O/B_I}},
\end{align}
\end{subequations}
where $\{\ket{\psi_i}^{A_I}\}$ is a set of random states in $\mathcal{L}(A_I)$ and $\{D_i^{A_O/B_I}\}$ is a set of optimization variables in $\mathcal{L}(A_O \otimes B_I)$.

We can also obtain the dual form of the inner approximation problems of robustnesses from the Lagrangian, which results in a similar SDP to Eqs.\,\eqref{eqs:dualgenrob} and Eqs.\,\eqref{eqs:dualwnrob}, but replacing Eqs.\,\eqref{subeq:witgenproj} and \eqref{subeq:witgendc} by

\begin{align}
\tr_{A_I}\left[ S (\ketbra{\psi_i}{\psi_i}^{A_I} \otimes \identity^{A_O B_I}) + \tracerep{B_I}{S_i} - \tracerep{A_O B_I}{S_i}\right] \succeq 0 \ \forall i \in \mathbb{N},
\end{align}

where $\{S_i\}$ is a set of variables in $\mathcal{L}(A_I \otimes A_O \otimes B_I)$.

\section{Heuristic algorithm that seeks for the maximum robustness on a given scenario}
\label{sec:seesaw}

	Given the dimensions $d_{A_I}, d_{A_O}, d_{B_I}$, what is the maximum robustness values $R_G(W)$ or $R_{WN}(W)$ that a process $W\in\mathcal{L}_{\AB}$ can obtain? Inspired by the see-saw techniques of Refs.\,\cite{Cavalcanti2017a, Bavaresco2017}, we now present a heuristic iterative method to seek for the highest robustness values for a given scenario. This method works either for the generalized or white noise robustness, which is why in the following we do not make explicit which robustness quantifier to work with.

Our heuristic algorithm works as follows. First, we sample a bipartite ordered process $W_1$ using one of the methods we describe in Appendix \ref{sec:random}.
	Then, we perform the dual robustness problem for $W_1$ and obtain its optimal non-classical CCDC witness $S_1$, which gives $R(W_1) = - \tr(S_1 W_1)$. Next, we find a process $W_2$ which maximally violates this first witness $S_1$, a problem that can be solved by the following SDP:

\begin{subequations}
	\begin{align}
	 & \textrm{min} \hspace{1.4cm}   \tr(S_1 W_2) \\
            & \textrm{s.t.}  \hspace{1.55cm}   W_2 \succeq 0, \\
            & \hspace{2cm} W_2 = L_{\AB}(W_2), \\
            & \hspace{2cm} \tr(W_2) = d_{A_O}.
	\end{align}
\end{subequations}

	Now we repeat the previous steps, that is, we evaluate the dual robustness program for $W_2$, and its optimal non-classical CCDC witness $S_2$, then find the process $W_3$ which maximally violates $S_2$. These steps are taken iteratively until some stopping criterion is satisfied. In our code, the stopping criterion used is $R(W_{i+1}) - R(W_{i}) \leq \epsilon = 0.0001$. In the end of this procedure, we obtain a process $W$ which attains a non-classical CCDC robustness, providing a lower bound on the maximal value for the given scenario. 
	
	In order to increase the confidence in this heuristic method, we perform this algorithm several times with for various  initial processes $W_1$ randomly sampled from different manners. We have implemented this heuristic methods for three different scenarios:
$d_{A_I}{=}d_{A_O}{=}d_{B_I}{=}2$, $d_{A_I}{=}d_{A_O}{=}2, d_{B_I}{=}4$, and $d_{A_I}{=}d_{A_O}{=}3, d_{B_I}{=}9$.
For these cases, our see-saw algorithm led to the same value of robustness for several different random initial processes, suggesting that the heuristic method may have attained the global maximum.

\section{Sampling random ordered process} \label{sec:random}

	In this section, we describe the methods for sampling random processes used in this work. We remark that Ref.\,\cite{Feix2016b} presents a method for generating uniformly distributed random process matrices which are different from the ones considered here.

\subsection*{Method 1}
	This method is similar to the technique for generating random quantum channels used in Ref.\,\cite{Bavaresco2020}.

\begin{enumerate}
	\item \label{subitem:randomprocess}
	Sort a random density operator $\rho \in \mathcal{L}(A_I \otimes A_O \otimes B_I)$ with the Hilbert-Schmidt measure, \emph{i.e}, sort a random pure quantum state with the Haar measure $\ketbra{\psi}{\psi}\in \mathcal{L}(A_I \otimes A_O \otimes B_I \otimes \aux)$ with $d_{\aux}=d_{A_I}d_{A_O}d_{B_I}$, then trace out the auxiliary space;
	\item Project $\rho$ into the subspace of bipartite ordered process to obtain $\overline{W} = L_{\AB}(\rho)$;
	\item Evaluate the minimum eigenvalue $\lambda_\text{min}$ of $\overline{W}$ and output the bipartite ordered process:
  		\begin{equation}
  		W = d_{A_O} \cdot \frac{\overline{W} - \lambda_\text{min}  \identity^{d_{A_I}d_{A_O}d_{B_I}}}{\tr\left(\overline{W} - \lambda_\text{min}  \identity^{d_{A_I}d_{A_O}d_{B_I}}\right)},
  		\end{equation}
 	which is positive semi-definite by construction.
 	\end{enumerate}

\subsection*{Method 2}
	This method is good for generating random processes with high values of generalized and white noise robustnesses. However, it only works when $\frac{d_{A_I}d_{A_O}}{d_{B_I}}$ is an integer;
 	\begin{enumerate}
 	
 	\item 
  Set $d_\aux=d_{A_I}$ and sort a random pure state $\ketbra{\psi}{\psi}^{A_I \aux} \in \mathcal{L}(A_I \otimes \aux)$ according to the Haar measure.
 	\item 
 	Set $d_{\aux'} = \frac{d_{A_O} d_\aux}{d_{B_I}}$ and sort a random unitary operator $\text{U}: A_O \otimes \aux \to B_I \otimes \aux'$ according to the Haar measure. Then, obtain its Choi operator $\keketbra{\text{U}}{\text{U}}^{A_O \aux/B_I \aux'}$ and define the channel $D^{A_O\aux / B_I} := \tr_{\aux'}\left(\keketbra{\text{U}}{\text{U}}^{A_O \aux/B_I \aux'}\right)$;
	 	
 	\item Output the operator:
	\begin{align}
		W & = \rho^{A_I \aux} * D^{A_O \aux/B_I },
	\end{align}
	which is a valid bipartite ordered process.	
\end{enumerate}

\subsection*{Method 3}

\begin{enumerate}
	\item Set $d_{\aux} = d_{A_I}$ and sort a random pure state $\ketbra{\psi}{\psi}^{A_I \aux}$ according to the Haar measure;
	\item Sort a random density operator $\rho \in \mathcal{L}(A_O \otimes \aux \otimes B_I)$ according to the Hilbert-Schmidt measure;
	\item Define $D^{A_O \aux/B_I} = \left( \sigma^{-\frac{1}{2}} \otimes \identity^{B_I}\right) \rho \left( \sigma^{-\frac{1}{2}} \otimes \identity^{B_I} \right)$, with $\sigma = \tr_{B_I} (\rho)$. The resulting operator $D^{A_O \aux/B_I}$ is the Choi operator of a channel $\map{D}: \mathcal{L}(A_I \otimes \aux) \rightarrow \mathcal{L}(B_I)$, as $D^{A_O \aux/B_I} \succeq 0$ and $\tr_{B_I} (D^{A_O \aux/B_I}) = \identity^{A_O \aux}$ by direct inspection;
	\item Output the operator
	\begin{equation}
		W = \ketbra{\psi}{\psi}^{A_I \aux} * D^{A_O \aux/B_I },
	\end{equation}
	which is a valid bipartite ordered process.
\end{enumerate}

\end{document}